\documentclass[12pt,a4paper]{amsart}

\usepackage{amsmath,amsfonts,amssymb,mathtools,extarrows}

\usepackage[hmargin=2cm,vmargin=2cm]{geometry}

\usepackage[hypertexnames=false,hyperfootnotes=false,colorlinks=true,linkcolor=blue,%
citecolor=purple,filecolor=magenta,urlcolor=cyan,unicode,linktocpage=true,pagebackref=false]{hyperref}

\usepackage{nameref,zref-xr}     

\usepackage{enumerate}      

\usepackage{scalerel}         
\usepackage{stmaryrd}

\usepackage[all]{xy}

\usepackage{amsmath,amsfonts,amssymb,amsthm,mathtools,MnSymbol}
\usepackage{graphicx}
\usepackage{subcaption}
\usepackage{amscd}
\usepackage{bbm}
\usepackage{enumerate}
\usepackage{galois}
\usepackage{mathrsfs}
\usepackage{xypic}
\usepackage{geometry}
\usepackage{hyperref}

\usepackage{color}

\newcommand{\oM}{\overline{\mathcal M}}

\def\oM{{\overline{\mathcal{M}}}}

\newcommand{\RT}{\mathrm{RT}}

\newcommand{\ZZ}{\mathbb{Z}}

\newcommand{\gl}{\mathrm{gl}}

\newcommand{\un}{{1\!\! 1}}

\newcommand{\Id}{\mathop{\mathrm{Id}}}
\newcommand{\nordbullet}{\mbox{\tiny ${\bullet\atop\bullet}$}}

\newcommand{\diff}[2]{\frac{\partial #1}{\partial #2}}

\newcommand{\qa}{\alpha}
\newcommand{\qb}{\beta}
\newcommand{\qd}{\delta}
\newcommand{\qg}{\gamma}
\newcommand{\qs}{\sigma}
\newcommand{\qt}{\tau}

\newcommand{\qe}{\varepsilon}
\newcommand{\qz}{\zeta}
\newcommand{\qp}{\partial}

\newcommand{\ql}{\lambda}

\newcommand{\Qo}{\Omega}
\newcommand{\mgn}{\overline{\mathcal M}_{g,n}}

\newcommand{\kk}[1]{\left(#1\right)}

\newcommand{\ddd}[1]{\llangle #1\rrangle_0}
\newcommand{\ddg}[1]{\llangle #1\rrangle}

\newtheorem{theorem}{Theorem}[section]
\newtheorem{proposition}[theorem]{Proposition}
\newtheorem{lemma}[theorem]{Lemma}
\newtheorem{corollary}[theorem]{Corollary}

\theoremstyle{remark}

\newtheorem{remark}[theorem]{Remark}

\theoremstyle{definition}

\newtheorem{definition}[theorem]{Definition}

\usepackage{xcolor}

\numberwithin{equation}{section}

\title[Universal identities]{Structure of Dubrovin-Zhang free energy functions and universal identities}

\author{Sergey Shadrin}
\address[S. Shadrin]{Korteweg--de Vries Instituut voor Wiskunde, Universiteit van Amsterdam, Postbus 94248, 1090GE Amsterdam, Nederland}
\email{s.shadrin@uva.nl}

\author{Zhe Wang}
\address[Z. Wang]{Division of Mathematics, Graduate School of Science, Kyoto University, Kyoto 606-8502, Japan}
\email{wang.zhe.65f@st.kyoto-u.ac.jp}

\begin{document}
	
	\begin{abstract}
		We prove a structural theorem relating the higher genera free energy functions of the  Dubrovin-Zhang hierarchies to the Witten-Kontsevich free energy function of the Korteweg-de Vries hierarchy. As an important application, for any given genus $g\geq 1$, we construct a set of universal identities valid for the free energy functions of any Dubrovin-Zhang hierarchy. In particular, we present some techniques that can be used to derive universal identities without relying on the geometry of the moduli space of stable curves of higher genus.
	\end{abstract}
	
	\maketitle
	
	\tableofcontents
	
	\section{Introduction}
	\label{ak}
	Since the proof of the Witten conjecture \cite{witten1990two} by Kontsevich \cite{kontsevich1992intersection}, which relates the topology of the moduli space of stable curves to the Korteweg-de Vries (KdV) hierarchy, people have gradually understood the deep relation between the 2D topological field theory and the theory of integrable hierarchies over past three decades. Many Witten-Kontsevich type theorems have been discovered and proved since then, where the corresponding integrable hierarchies (which were already known to mathematical physicists in completely different contexts) appear to universally govern partition functions constructed from different aspects of mathematical physics, for example, from quantum cohomology, matrix model, singularity theory, etc., {see \cite{bessis1980quantum,carlet2004extended,getzler2001toda,givental2001gromov,givental2005simple,t1973planar,witten1990two} and references therein.}
	To systematically study the emerging integrable hierarchies, Dubrovin and Zhang started a program in \cite{dubrovin2001normal} aiming at giving an  axiomatic characterization of topological integrable hierarchies, that is, the integrable evolutionary PDEs that control a certain  2D topological field theory.
	
    The paper \cite{dubrovin2001normal} has two main goals. The first one is to construct an integrable hierarchy from a given 2D topological field theory, and the second goal is to reproduce all the universal identities satisfied by all Gromov-Witten invariants (and all other enumerative invariants that fit this context) at full genera. Up to now, the first goal is completely achieved under the semisimplicity assumption. To state precisely, given a 
    semisimple Frobenius manifold with a calibration, there exists a unique tau-symmetric bihamiltonian integrable hierarchy, called the Dubrovin-Zhang (DZ) hierarchy, whose tau-function satisfies a family of linear Virasoro constraints. As for the second goal, Dubrovin and Zhang derive
    the topological recursion relations for genus zero and one from the construction of DZ hierarchies. Remarkably, those topological recursion relations are originally proved by studying the topology of the moduli space of stable curves. However, Dubrovin and Zhang's method indicates that those relations are hidden in every DZ hierarchies, even for those hierarchies that have no known relations to geometric enumerative problems or topology of moduli spaces. More generally, Dubrovin and Zhang prove in~\cite{dubrovin2005normal} that the tau functions of their hierarchies are given by the Givental formula~\cite{givental2001gromov,givental2001semisimple}, which is known to satisfy all universal relations coming from the relations among additive generators of the tautological ring of the moduli spaces of curves~\cite{faber2010tautological}. 
    
    In some cases it is clear what structural property  a Dubrovin-Zhang hierarchy {possesses} if {it is} governed by a particular set of tautological relations, see e.~g.~\cite{buryak2012polynomial,iglesias2022bi}. But it is largely unknown for more involved types of relations, and the whole theory lacks simple explicitly written universal identities that would reflect some explicitly understood universal properties of the Dubrovin-Zhang hierarchies and/or their tau-functions. 
    In particular, for practical applications it is useful to have universal identities that provide efficient tools to control the ingredients of the Dubrovin-Zhang hierarchies in terms of the dependent variables and their derivatives with respect to the spacial variable (the so-called jet variables). For instance, the constraints on the free energy functions written in terms of the Eguchi-Xiong operators~\cite{eguchi1998quantum} imply the so-called $3g-2$ property of the free energy function, cf.~\cite{buryak2012polynomial}.

    In this paper, we make a further step towards studying the universal identities by using the method of Dubrovin and Zhang. For a given genus $g\geq 1$, we construct a set of universal identities valid for the topological solution of the Dubrovin-Zhang hierarchy of any semisimple Frobenius manifold, and the operators that we involve have strong vanishing properties in the jet variables. These operators naturally generalize the Eguchi-Xiong operators to higher differential order.   
	
	The key ingredient for deriving these universal identities is a structural theorem for the free energy functions. To state the result precisely, recall that for 
	a semisimple Frobenius manifold with a fixed calibration, the higher genus free energy functions $\mathcal F_g$ can be written as functions depending on jet variables for $g\geq 1$. For example, consider the following free energy function of the 
	Gromov-Witten theory of the point:
	\[
		\mathcal F_g = \sum_{n\geq 0}\sum_{k_1,\dots,k_n\geq 0}\frac{t_{k_1}\dots t_{k_n}}{n!} \int_{\mgn}\psi_1^{k_1}\dots\psi_n^{k_n}.
		\]
	It is well-known that for $g\geq 1$, we have 
	\begin{equation}
		\label{ab}
		\mathcal F_g = F_g\left(u^{(1)},\dots,u^{(3g-2)}\right),
	\end{equation}
	here 
	\[
		u^{(s)} = \frac{\qp^{s+2} \mathcal F_0}{\qp t_0^{s+2}},\quad s\geq 1,
		\]
	and
	\[
		F_1 = \frac{1}{24}\log u^{(1)},\quad F_g\in \mathbb Q \left[\frac{1}{u^{(1)}},u^{(2)},\dots,u^{(3g-2)}\right],\quad g\geq 2.
		\]
For example, we have 
\[
	F_2 = \frac{u_{xx}^3}{360u_x^4}-\frac{7u_{xx}u^{(3)}}{1920u_x^3}+\frac{u^{(4)}}{1152 u_x^2};
	\]
here we use the notation $u_x = u^{(1)}$ and $u_{xx} = u^{(2)}$. Note that the 
Gromov-Witten theory of the point corresponds to the one-dimensional Frobenius manifold given by the potential
\[
	F = \mathcal F_0|_{t_0 = u, t_{>0} = 0} = \frac{u^3}{6}.
	\]
Generally speaking, 
if the underlying semisimple Frobenius manifold $M$ {is of dimension} $N$ and its first metric in the flat coordinates $v^1,\dots,v^N$ is denoted by $\eta$,
then its higher genera free energy functions $\mathcal F_g$ have the form
\[
	\mathcal F_g = F_g\left(v^\qa,v^{\qa,1},\dots,v^{\qa,3g-2}\right),
	\]
here \[
	v^\qa = \eta^{\qa\qb}\frac{\qp^2\mathcal F_0}{\qp t^{\qb,0}\qp t^{1,0}},\quad v^{\qa,s} = \left(\diff{}{t^{1,0}}\right)^s v^\qa,\quad s\geq 1,\quad \qa = 1,\dots,N,
	\]
and
\[
	F_1 = \frac{1}{24}\log\det \left(c_{\qa\qb\qg}v^{\qg,1}\right)+G(v),\quad F_g\in C^{\infty}(v)\left[\frac{1}{v^{\qa,1}},v^{\qa,1},v^{\qa,2},\dots,v^{\qa,3g-2}\right]
,\quad g\geq 2.	\]
Here and henceforth, we will always assume the Einstein summation rule for upper and lower Greek indices. Moreover, we will always raise or lower indices by the metric $\eta$. Note that in the above expressions, $G(v)$ is the so-called $G$-function defined in \cite{getzler1997intersection} (see also \cite{dubrovin1998bihamiltonian}), $c_{\qa\qb\qg}$ are functions defined by
\[
	c_{\qa\qb\qg} =\left.\frac{\qp^3 \mathcal F_0}{\qp t^{\qa,0}\qp t^{\qb,0}\qp t^{\qg,0}}\right|_{t^{\qa,0} = v^\qa,t^{\qa,1} = t^{\qa,2}=\dots=0},
	\]
and $C^{\infty}(v)$ is the ring of smooth function depending on variables $v^1,\dots,v^N$. Furthermore, {let us} denote by $(u^1,\dots,u^N)$ the canonical coordinates on $M$, {then} in terms of these coordinates, the metric $\eta$ is of the diagonal form
\[
	\eta = \sum_{i=1}^N f_i(u) (du^i)^2.
	\]
Each $u^i$ can be viewed as functions depending on $v^1,\dots,v^N$, therefore we see that $F_g$ for $g\geq 1$ can also be written as functions in jet variables of $u^i$, that is, $F_g$ can be viewed as a function in $u^i,u^{i,1},\dots,u^{i,3g-2}$, with
\[
u^{i,s} = \left(\diff{}{t^{1,0}}\right)^s u^i,\quad s\geq 1,\quad i = 1,\dots,N.
	\]
Now we can state the following structural theorem for $F_g$.

\begin{theorem}
	\label{ae}
Given 
a semisimple Frobenius manifold (with a choice of calibration) of rank $N$, its higher genus free energy function $F_g$ admits the following decomposition:
\begin{equation}
	\label{ah}
	F_g = \sum_{i=1}^N\left(\frac{1}{f_i(u)}\right)^{g-1}F_g^{KdV}(u^{i,1},\dots,u^{i,3g-2})+H_g,\quad g\geq 1,
\end{equation}
	here $F_g^{KdV}$ is the genus $g$ free energy function \eqref{ab} of the Gromov-Witten theory of the point and the function $H_g$ satisfies the conditions
\begin{equation}\label{eq:RestrictionOnH}
		\frac{\qp^n H_g}{\qp u^{i_1,k_1}\dots\qp u^{i_n,k_n}} = 0,\quad k_1+\dots+k_n\geq 3g-3+n,\quad n\geq 1.
\end{equation}
\end{theorem}

Using the above theorem, we derive a family of universal identities. To this end, we define and study thoroughly a set of differential operators $O_{\{\alpha_1,k_1;\dots,\alpha_n,k_n\}}$, where $n = k_1 = 1$ or 
\[
	k_1+\cdots +k_n =   3g-3 +n,\quad g\geq 2,\quad n\geq 1,\quad k_i\geq 2.
	\]
These operators are differential operators of degree $n$ on the large phase space and have strong vanishing properties in the jet coordinates that allow us to use equation~\eqref{eq:RestrictionOnH} to annihilate $H_g$. 

Introduce the following correlators
\[
	\llangle\qt_{\qa_1,k_1}\dots\qt_{\qa_n,k_n}\rrangle_g:=\frac{\qp^n \mathcal F_g}{\qp t^{\qa_1,k_1}\dots\qp t^{\qa_n,k_n}}
	\]
to express the action $O_{\{\alpha_1,k_1;\dots,\alpha_n,k_n\}}(\mathcal F_g)$. We then have the following theorem on universal relations.

\begin{theorem}
	\label{af}
	Given 
	a semisimple Frobenius manifold (with a choice of calibration), its genus $g$ correlators satisfy the relations
\begin{equation} \label{eq:UniversalIdentity-intro}
		O_{\{\alpha_1,k_1;\dots,\alpha_n,k_n\}}(\mathcal F_g) = B^g_{k_1,\dots,k_n}M[g]^{\qg_0}_{\qg_n}\prod_{i=1}^n\ddd{\qt_{\qa_i,0}\qt_{\qg_{i-1,0}}\qt^{\qg_{i},0}},
\end{equation}
for $g\geq 1$. In the above expressions, $B^g_{k_1,\dots,k_n}$ are some rational numbers that can be explicitly computed from the intersection numbers of $\mgn$, and $M[g]$ is defined by 
\[
	M[g]^\qa_\qg =\begin{cases}
		\qd^\qa_\qg, & \text{for } g=1, \\
		M^\qa_\qg, & \text{for } g=2, \\
		M^{\qa}_{\qb_1}M^{\qb_1}_{\qb_2}\dots M^{\qb_{g-3}}_{\qb_{g-2}}M^{\qb_{g-2}}_{\qg}, & \text{for } g\geq 3,
	  \end{cases}
	\]
where we denote 
\[
M^{\qa}_\qb = \ddd{\qt_{0}^\qa\qt_{\ql,0}\qt_{\mu,0}}\ddd{\qt_{0}^\ql\qt_{0}^\mu\qt_{\qb,0}}.
	\]
\end{theorem}

Let us give some examples of these relations. The simplest ones are given by the operators $O_{\{\qa,3g-2\}}$, which coincide with the Eguchi-Xiong differential operators introduced in \cite{eguchi1998quantum} (see also \cite{buryak2012polynomial,liu2002quantum}), and the corresponding relations read
\[
	O_{\{\qa,3g-2\}}(\mathcal F_g) = \ddd{\qt_{\qa,0}\qt_{\qb,0}\qt^{\ql}_0}M[g]^\qb_\ql\int_{\oM_{g,1}}\psi_1^{3g-2},\quad g\geq 1.
	\]
For $g = 1$ the operator reads
\[
	O_{\{\qa,1\}} = \diff{}{t^{\qa,1}}-\ddd{\qt_{\qa,0}\qt^\qb_0}\diff{}{t^{\qb,0}},
	\] 
and we obtain the well-known relation~\cite{dijkgraaf1990mean}
\[
\ddg{\qt_{\qa,1}}_1-\ddd{\qt_{\qa,0}\qt^\qb_0}\ddg{\qt_{\qb,0}}_1 = \frac{1}{24}\ddd{\qt_{\qa,0}\qt^\qb_0\qt_{\qb,0}}.
	\]
For $g = 2$, we can derive three relations from Theorem~\ref{af}, and they read
\begin{align*}
O_{\{\qa,4\}}(\mathcal F_2) &= \frac{1}{1152}\ddd{\qt_{\qa,0}\qt_{\ql,0}\qt^\mu_0}M^\ql_\mu,\\
O_{\{\qa,3;\qb,2\}}(\mathcal F_2)&=-\frac{7}{1920}\ddd{\qt_{\qa,0}\qt^\qg_0\qt_{\ql,0}}\ddd{\qt_{\qb,0}\qt_{\qg,0}\qt^\mu_0}M^\ql_\mu,\\
O_{\{\qa,2;\qb,2;\qg,2\}}(\mathcal F_2)&=\frac{1}{60}\ddd{\qt_{\qa,0}\qt^\qe_0\qt_{\ql,0}}\ddd{\qt_{\qb,0}\qt_{\qe,0}\qt^\qs_0}\ddd{\qt_{\qg,0}\qt_{\qs,0}\qt^\mu_0}M^\ql_\mu.\\
\end{align*}

It is important to stress that the main statement is not just the shape of the identities in Theorem~\ref{af} or the particular formula for the right hand side in Equation~\eqref{eq:UniversalIdentity-intro}. In fact, one can easily produce alternative versions of such identities using the lifts of the vanishing tautological classes of high enough degree --- the idea that we mentioned above. Moreover, it is exactly what we have there for $n=1$: we obtain the Eguchi-Xiong vector fiels on the left hand side from the tautological relation that identifies the top nontrivial degree of the so-called $\psi$-class with the class of the point. However, starting from $n=2$ we deviate from this naive approach. And the really interesting part of the statement that we propose is the strong vanishing properties of the operators on the left hand side. In particular, they pick up information just from the leading terms with respect to the $\partial_x$-degree that is concentrated in $F_g^{KdV}$ part of Equation~\eqref{ah}. 

\subsection*{Organization of the paper}
This paper is organized as follows. In Sect.\,\ref{ac} we prove Theorem \ref{ae} by combining the Givental's quantization formalism and Dubrovin-Zhang's loop equation approach. In Sect.\,\ref{ad} we give a detailed description of the operators $O_{\{\alpha_1,k_1;\dots,\alpha_n,k_n\}}$ and study their properties, then we prove Theorem \ref{af}. {We also discuss possible approaches for deriving more  general universal identities.} In Sect.\,\ref{ag}, we give some concluding remarks.

	\section{Structure of free energy functions} 
	\label{ac}
	In this section, we prove Theorem \ref{ae}. The idea is to prove that the decomposition \eqref{ah} is invariant under the Givental's twisted loop group action. 
	{Then Theorem \ref{ae} follows from the fact that the tau-function of any calibrated semisimple Frobenius manifold can be computed via Givental's group actions from the Witten-Kontsevich tau-function, for which Theorem \ref{ae} holds true trivially.} A similar idea has been  used to prove the invariance of tautological equations \cite{faber2010tautological,lee2009invariance} and used to prove the polynomiality property of DZ hierarchies \cite{buryak2012deformations,buryak2012polynomial}. 
	\subsection{Givental theory}
In \cite{givental2001gromov,givental2001semisimple,givental2004symplectic}, Givental introduced a twisted loop group action on the space of tame partition functions. Teleman proved~\cite{teleman2012structure} that partition functions of all semisimple cohomological field theories wih the same underlying Frobenius algebra structure lie in the same orbit of the group action. In this subsection, let us recall the basic formalism, one may refer to, e.g., \cite{buryak2012polynomial,faber2010tautological,lee2009invariance} for expositions.

Let $H$ be an $N$-dimensional vector space equipped with a non-degenerate bilinear pairing $\langle-,-\rangle$. Consider the space $\mathcal H = H\otimes\mathbb C((z^{-1}))$ together with the bilinear map
\[
	\Qo(f,g) = \frac{1}{2\pi i}\int \langle f(-z),g(z)\rangle dz,\quad f,g\in\mathcal H,
	\]
one can show that this is a symplectic form and $\mathcal H$ is called the Givental symplectic space. Let $M$ be a symplectomorphism of $\mathcal H$ of the form
\[
	M = \sum_k M_kz^k,\quad M_k\in\mathrm{End}(H).
	\]
The action of $M$ on tame partition functions is denoted by $\hat M$ and is given by the exponential of the action of the corresponding Lie algebra element, that is, if we write $M = \exp(m)$, where $m$ is an infinitesimal symplectic transformation, then $\hat M:=\exp(\hat m)$. The action $\hat m$ is then given by the standard Weyl quantization of the quadratic Hamiltonian
\[
	h_m(f) = \frac12 \Qo\left(m(f),f\right).
	\]

For our purpose, we will only consider the action of upper triangular elements and lower triangular elements.
An infinitesimal symplectic transformation $m$ is of the form 
\[
	m = \sum_k m_kz^k,\quad m_k\in\mathrm{End}(H)
	\]
and satisfies the condition
\[
	\Qo(m(f),g)+\Qo(f,m(g)) = 0,\quad f,g\in\mathcal H.
	\]
Such a transformation $m$ is called upper triangular if $m_k = 0$ for $k\leq 0$ and called lower triangular if $m_k = 0$ for $k\geq 0$. To write down the explicit expressions of the upper and lower triangular action, let us fix an orthonormal basis $e_1,\dots,e_N$ of $H$ and denote $\un = e_1+\dots+e_N$. This basis determines coordinates $t^{i,k}$ of the large phase space, where $i = 1,\dots, N$ and $k\geq 0$. A tame partition function is the exponential of a formal power series of the form
\[
	Z = \exp\left(\sum_{g\geq 0}\qe^{2g-2}\mathcal F_g\right),\quad \mathcal F_g\in\mathbb C[[t^{i,k}]] 
	\]
that satisfies certain properties. Then an upper triangular transformation
\[
	\mathfrak r = \sum_{k\geq 1} \mathfrak r_kz^k,\quad \mathfrak r_k\in\mathrm{End}(H)
	\]
acts on a tame partition function $Z$ by
\begin{equation}
	\label{ai}
	\hat{\mathfrak r}[t].Z = \hat{\mathfrak r}Z,
\end{equation}
where $\hat{\mathfrak r}$ is the following second order differential operator on the large phase space:
\begin{align*}
	\hat{\mathfrak r} =& -\sum_{k\geq 1}\sum_{i=1}^N (\mathfrak r_k)^i_{\un}\diff{}{t^{i,k+1}}+\sum_{\ell\geq 0,k\geq 1}\sum_{i,j=1}^N (\mathfrak r_k)^i_{j}t^{j,\ell}\diff{}{t^{i,\ell+k}}\\
	&+\frac{\qe^2}{2}\sum_{k,\ell\geq 0}\sum_{i,j=1}^N (-1)^{k+1} (\mathfrak r_{k+\ell+1})^{i,j}\frac{\qp^2}{\qp t^{i,k}\qp t^{j,\ell}}.
\end{align*}
Note that the action \eqref{ai} induces the action
\begin{equation}
	\label{an}
	\hat{\mathfrak r}[t].\mathcal F = Z^{-1}\hat{\mathfrak r}Z
\end{equation}
on the free energy function $\mathcal F = \log Z$. Similarly, a lower triangular transformation
\[
	\mathfrak s = \sum_{k\geq 1} \mathfrak s_kz^{-k},\quad \mathfrak s_k\in\mathrm{End}(H)
	\]
acts by the following first order differential operator:
\begin{align*}
	\hat{\mathfrak s} =& -\frac{1}{2\qe^2}(\mathfrak s_3)_{\un,\un}+\frac{1}{\qe^2}\sum_{k\geq 0}\sum_{i=1}^N (\mathfrak s_{k+2})_{\un,i}t^{i,k}+\frac{1}{2\qe^2}\sum_{k,\ell\geq 0}\sum_{i,j=1}^N (-1)^i(\mathfrak s_{k+\ell+1})_{i,j}t^{i,k}t^{j,\ell}\\
	&-\sum_{i=1}^N(\mathfrak s_1)^i_{\un}\diff{}{t^{i,0}}+\sum_{k\geq 0,\ell\geq 1}\sum_{i,j=1}^N(\mathfrak s_\ell)^i_j t^{j,k+\ell}\diff{}{t^{i,k}}.
\end{align*}

The twisted loop group action described above in particular allows one to reconstruct the partition function of any \emph{homogeneous} semisimple cohomological field theory from its genus zero data~\cite{teleman2012structure}. Let $(V,\eta,\{c_{g,n}\})$ be a semisimple cohomological field theory of rank $N$, that is, $V$ is an $N$-dimensional vector space with a non-degenerate bilinear form $\eta$, $c_{g,n}$ are families of multilinear maps
\[
	c_{g,n}\colon  V^{\otimes n}\to H^*(\mgn,\mathbb C),\quad 2g-2+n\geq 0
	\]
that satisfy certain properties. The partition function associated to this cohomological field theory is defined to be 
\begin{align*}
	Z &= \exp\left(\sum_{g\geq 0}\qe^{2g-2}\mathcal F_g\right),\\ \mathcal F_g &= \sum_{\substack{n\geq 0\\ 2g-2+n>0}}\sum_{k_1,\dots,k_n\geq 0}\frac{t^{\qa_1,k_1}\dots t^{\qa_n,k_n}}{n!}\int_{\mgn}c_{g,n}(e_{\qa_1}\otimes\dots\otimes e_{\qa_n})\psi_1^{k1}\dots\psi_n^{k_n},
\end{align*}
	here $e_1,\dots, e_N$ is a fixed basis of $V$ with $e_1$ being the unit of the theory, and $\psi_i$ is the first Chern class of the $i$-th tautological line bundle of $\mgn$.
Then, in the homogeneous case, Givental \cite{givental2001gromov,givental2001semisimple} proposed the following formula: 
\begin{equation}
	\label{aj}
	Z = C_{pt}\hat S^{-1}_{pt} \hat\Psi_{pt}\hat R_{pt}  \prod_{i=1}^N Z^{KdV}\left(\frac{\qe^2}{f_i};\frac{T^{i,0}}{\sqrt{f_i}},\frac{T^{i,1}}{\sqrt{f_i}},\dots\right),
\end{equation}
which was identified with the tau-function of the Dubrovin-Zhang hierarchy of the underlying Frobenius manifold with an appropriate choice of calibration by Dubrovin and Zhang~\cite[Theorem 4.3.14]{dubrovin2005normal} and proved to hold for any homogenenous semisimple cohomological field theory by Teleman~\cite{teleman2012structure}.
 
Let us explain the notation in the above formula, and one may refer to \cite{givental2001gromov} for details. The function $Z^{KdV}$ is the Witten-Kontsevich tau-function of the KdV hierarchy, namely it is given by
\begin{align*}
	Z^{KdV}\left(\qe^2; T^{i,0},T^{i,1},\dots\right) &= \exp\left(\sum_{g\geq 0}\qe^{2g-2}\mathcal F_g^{KdV}(T^{i,0},T^{i,1},\dots)\right)\\
	\mathcal F_g^{KdV}\left(T^{i,0},T^{i,1},\dots\right)& = \sum_{n\geq 0}\sum_{k_1,\dots,k_n\geq 0}\frac{T^{i,k_1}\dots T^{i,k_n}}{n!} \int_{\mgn}\psi_1^{k_1}\dots\psi_n^{k_n}.
\end{align*}
The upper triangular symplectic transformation $R_{pt}$ and the lower triangular one $S_{pt}$ are both determined from the  underlying semisimple Frobenius manifold $M$ corresponding to the given homogenenous cohomological field theory. Note that these two transformations vary on $M$, and we fix them by taking their values at an arbitrary (semisimple) point $pt\in M$. In a neighborhood of $pt$, we denote by $(v^1,\dots,v^N)$ the flat coordinates of $M$ corresponding to the basis $(e_1,\dots,e_N)$ and by $(u^1,\dots,u^N)$ the canonical coordinates. It is well-known that in terms of the canonical coordinates the metric $\eta$ is diagonal whose diagonal elements we denote by $f_i$. We then define the matrix $\Psi$ to be 
\[
	\Psi_{i\qa} = \sqrt{f_i}\ \diff{u^i}{v^\qa},
	\]
and the matrix $\Psi_{pt}$ is obtained by evaluating functions $\Psi_{i\qa}$ at the point $pt$. The transformation $\hat\Psi$ is then a coordinate transformation from the normalized canonical time variables $T^{i,k}$ to the flat time variables $t^{\qa,p}$. Finally, $C$ is just a function on $M$, and we denote by $C_{pt}$ its value at the point $pt$.

\subsection{Loop equation of the free energy function}
\label{al}
To prove Theorem~\ref{ae}, we recall in this subsection the Dubrovin-Zhang's loop equation method \cite{dubrovin2001normal} for computing the free energy function.

{Let $M$ be a semisimple Frobenius manifold of dimension $N$ with a fixed calibration, denote its flat coordinates by $v^1,\dots,v^N$ and the first flat metric by $\eta$.} 
As we have introduced in Sect.\,\ref{ak}, the higher genus free energy functions can be written as functions in the jet coordinates $v^{\qa,s}$ of $M$ \cite{buryak2012polynomial,dubrovin2001normal}, where $v^\qa$ are flat coordinates of $M$, or equivalently as functions in $u^{i,s}$ where $u^i$ are canonical coordinates of $M$:
\[
	\mathcal F_g(t^{\qa,p}) = F_g(u^{i},u^{i,1},\dots,u^{i,3g-2}),\quad g\geq 1
	\]
where we view $u^{i}$ as functions of $v^\qa$ and 
\[
	v^\qa = \eta^{\qa\qb}\frac{\qp^2\mathcal F_0}{\qp t^{\qb,0}\qp t^{1,0}},\quad v^{\qa,s} = \left(\diff{}{t^{1,0}}\right)^s v^\qa,\quad u^{i,s} = \left(\diff{}{t^{1,0}}\right)^s u^i,\quad s\geq 1.
	\]
We will also use the notation $v^{\qa,0} = v^\qa$ and $u^{i,0} = u^i$. In \cite{dubrovin2001normal}, Dubrovin and Zhang give a way to uniquely reconstruct $F_g$ from $M$ by requiring the so-called linearized Virasoro constraints, and they derive the following \emph{loop equation} satisfied by $F_g$ for $g\geq 1$:
\begin{align}
	\label{ao}
	&\sum_{r\geq 0}\diff{F_g}{v^{\qg,r}}\qp_x^r\kk{\frac{1}{E-\ql}}^\qg+\sum_{r\geq 1}\diff{F_g}{v^{\qg,r}}\sum_{k=1}^r\binom{r}{k}\qp_x^{k-1}\qp_1p_\qa G^{\qa\qb}\qp_x^{r-k+1}\qp^\qg p_\qb\\
	\notag
	=&\,\frac{1}{2}\sum_{k,\ell\geq 0}\kk{\sum_{m = 1}^{g-1}\diff{F_{m}}{v^{\qg,k}}\diff{F_{g-m}}{v^{\rho,\ell}}+\frac{\qp^2 F_{g-1}}{\qp v^{\qg,k}\qp v^{\rho,\ell}}}\qp_x^{k+1}(\qp^\qg p_\qa) G^{\qa\qb}\qp_x^{\ell+1}(\qp^\rho p_\qb)\\
	\notag
	&+\frac{1}{2}\sum_{k\geq 0}\diff{F_{g-1}}{v^{\qg,k}}\qp_x^{k+1}\left[\nabla\diff{p_\qa}{\ql}\cdot \nabla\diff{p_\qb}{\ql}\cdot v_x\right]^\qg G^{\qa\qb}+h(v,\ql)\qd_{g,1};
	\end{align}
here and henceforth we use the notation
\[
\qp_x = \diff{}{t^{1,0}},\quad \qp_\qa = \diff{}{v^\qa},\quad \qp^\qa = \eta^{\qa\qb}\diff{}{v^\qb},
	\]
and on the right-hand side we set $F_0:=0$.

Let us explain how $F_g$ is obtained from the above equation. Note first that $\ql$ appeared in the equation \eqref{ao} is a formal parameter, and we solve the equation with respect to  $F_g$ such that \eqref{ao} holds true for any $\ql$. To make it precise, it is proved that when written in the canonical coordinates, the left-hand side is of the form
\[
	\sum_{r\geq 0}\sum_{i=1}^N\diff{F_g}{u^{i,r}}K^{i,r},\quad  K^{i,r}\in \mathcal A\left[\frac{1}{\ql-u^1},\dots,\frac{1}{\ql-u^N}\right],
	\]
where $\mathcal A$ is the ring of differential polynomial given by 
\[
	\mathcal A = C^{\infty}(u)[u^{i,s}\colon s\geq 1,i = 1,\dots, N].
	\]
It is also proved that (see Lemma 3.10.19 of \cite{dubrovin2001normal}), when viewed as a polynomial in $\frac{1}{\ql-u^1},\dots,\frac{1}{\ql-u^N}$, each $K^{i,r}$ is of degree $r+1$ of the form
\begin{equation}
	\label{ap}
	K^{i,r} = \frac{g^{i,r}}{(\ql-u^i)^{r+1}}+\text{lower order terms},\quad g^{i,r}\in\mathcal A,\quad  g^{i,r}\neq 0.
\end{equation}
On the right-hand side, the function $h(v,\ql)$ can be expressed in terms of the canonical coordinates by
\[
	h(v,\ql) = -\frac 18\sum_{i=1}^N \frac{1}{(\ql-u^i)^2}+\sum_{\substack{i,j = 1,\dots N\\i<j}}h_{ij}(u)\left(\frac{1}{\ql-u^i}-\frac{1}{\ql-u^j}\right).
	\]
Therefore, for $g = 1$, we have the equation
\[
	\sum_{r\geq 0}\sum_{i=1}^N\diff{F_1}{u^{i,r}}K^{i,r} = -\frac 18\sum_{i=1}^N \frac{1}{(\ql-u^i)^2}+\sum_{\substack{i,j = 1,\dots N\\i<j}}h_{ij}(u)\left(\frac{1}{\ql-u^i}-\frac{1}{\ql-u^j}\right),
	\]
from which one observes that
\[
	\diff{F_1}{u^{i,r}} = 0,\quad r\geq 2,\quad i = 1,\dots, N,
	\]
and obtains the well-known $g=1$ free energy function (\cite{dijkgraaf1990mean,getzler1997intersection,dubrovin1998bihamiltonian}, see also Sect.\,3.10.7 of \cite{dubrovin2001normal} for a detailed discussion)
\[
	F_1 = \frac{1}{24}\log\det \left(c_{\qa\qb\qg}v^{\qg,1}\right)+G(v),
	\]
where $G(v)$ is Getzler's $G$-function. For $g\geq 2$, we see that the right-hand side of the loop equation \eqref{ao} only depends on $F_1,\dots,F_{g-1}$, so we can find recursively all $F_g$ starting from $F_1$. Moreover, it is proved that the right-hand side is a function in the ring
\[
	\mathcal A\left[\frac{1}{u^1\dots u^N}\right]\left[\frac{1}{\ql-u^1},\dots,\frac{1}{\ql-u^N}\right],
	\]
therefore one finds the gradient $\diff{F_g}{u^{i,r}}$ by comparing the coefficients of monomials in $\frac{1}{\ql-u^1},\dots,\frac{1}{\ql-u^N}$.

Let us proceed to explain the notations in the loop equation \eqref{ao}. Recall that the derivative $\qp_x$ is just $\diff{}{t^{1,0}}$. $E$ is the Euler vector field of $M$ and { 
	\[
\kk{\frac{1}{E-\ql}}^\qg = \sum_{m\geq -1}\frac{1}{\ql^{m+2}}\kk{E^{m+1}}^\qg,\quad E^{m+1} = E^m\cdot E,\quad E^{0}:=e,
\]
where $\cdot$ is the quantum product on $TM$ and $e$ is the unit vector field with respect to the quantum product.} The functions $p_\qa(v;\ql)$ are so-called periods of $M$, which are solutions of the Gauss-Manin system associated to $M$ (\cite{dubrovin1996geometry}, see also Sect.\,3.6.3 of \cite{dubrovin2001normal}), and $G^{\qa\qb}$ are some constants where the matrix $(G^{\qa\qb})$ is the Gram matrix of the flat pencil of $M$ with respects to the periods. In the expression
\[
	\left[\nabla\diff{p_\qa}{\ql}\cdot \nabla\diff{p_\qb}{\ql}\cdot v_x\right]^\qg,
	\]
$\nabla$ is the Levi-Civita connection of the flat metric $\eta$ and $v_x$ is the vector with components $v^{\qa,1}$. Therefore, we see that 
\begin{equation}
	\label{be}
	\left[\nabla\diff{p_\qa}{\ql}\cdot \nabla\diff{p_\qb}{\ql}\cdot v_x\right]^\qg = \left(\frac{\qp^2 p_\qa}{\qp\ql\qp v^{\qz}}\right)\left(\frac{\qp^2 p_\qb}{\qp\ql\qp v^{\mu}}\right)c^{\qd}_{\qz\mu}c^{\qg}_{\qd\qe}v^{\qe,1},
\end{equation}
where $c^{\qa}_{\qb\qg}$ are the structure constants of the quantum product given by
\[
	c^{\qa}_{\qb\qg} = \eta^{\qa\mu}\left.\frac{\qp^3\mathcal F_0}{\qp t^{\mu,0}t^{\qb,0}t^{\qg,0}}\right|_{t^{\qa,0} = v^\qa, t^{\qa,1} = t^{\qa,2} =\dots=0}.
	\]

As an example, let us consider the loop equation of the Gromov-Witten theory of the point. The flat coordinate of $M$ is $v^1$, and it is also the canonical coordinate, $u^1 = v^1$. The loop equation for $F_g^{KdV}$ is given by
\begin{align}
	\label{aq}
	&\sum_{r\geq 0}\diff{F_g^{KdV}}{v^{1,r}}\qp_x^r\kk{\frac{1}{v^1-\ql}}+\sum_{r\geq 1}\diff{F_g^{KdV}}{v^{1,r}}\sum_{k=1}^r\binom{r}{k}\qp_x^{k-1}\kk{\frac{1}{\sqrt{v^1-\ql}}}\qp_x^{r-k+1}\kk{\frac{1}{\sqrt{v^1-\ql}}}\\
	\notag
	=&\,\frac{1}{2}\sum_{k,\ell\geq 0}\kk{\sum_{m = 1}^{g-1}\diff{F_{m}^{KdV}}{v^{1,k}}\diff{F_{g-m}^{KdV}}{v^{1,\ell}}+\frac{\qp^2 F_{g-1}^{KdV}}{\qp v^{1,k}\qp v^{1,\ell}}}\qp_x^{k+1}\kk{\frac{1}{\sqrt{v^1-\ql}}} \qp_x^{\ell+1}\kk{\frac{1}{\sqrt{v^1-\ql}}}\\
	\notag
	+&\frac{1}{8}\sum_{k\geq 0}\diff{F_{g-1}}{v^{1,k}}\qp_x^{k+1}\kk{\frac{v^{1,1}}{(v^1-\ql)^3}}-\frac{\qd_{g,1}}{16(v^1-\ql)^2},
	\end{align}
which we can solve with respect to $F_g^{KdV}$, and obtain
\[
	F_1^{KdV} = \frac{1}{24}\log v^{1,1},\quad F_2^{KdV} = \frac{(v^{1,2})^3}{360(v^{1,1})^4}-\frac{7v^{1,2}v^{1,3}}{1920(v^{1,1})^3}+\frac{v^{1,4}}{1152 (v^{1,1})^2},\quad\dots.
	\]
To prepare for the proof of Theorem \ref{ae}, let us analyze the loop equation \eqref{aq} for $g\geq 2$ and reproduce some well-known results about the structure of the $F_g^{KdV}$. 

We first observe that both sides of the equation are homogeneous with respect to the differential degree $\deg_{\qp_x}$ defined by
\[
	\deg_{\qp_x}v^1 = 0,\quad \deg_{\qp_x}v^{1,k} = k,\quad k\geq 1.
	\]
Hence, by induction on $g$, we see that $\deg_{\qp_x}F_{g}^{KdV} = 2g-2$. Moreover, by a straightforward computation we find that, viewed as a polynomial in $\frac{1}{\ql-v^1}$,
\begin{equation}
	\label{at}
	\qp_x^{n_1}\kk{\frac{1}{\sqrt{v^1-\ql}}} \qp_x^{n_2}\kk{\frac{1}{\sqrt{v^1-\ql}}} = A_{n_1,n_2}\frac{(v^{1,1})^{n_1+n_2}}{(\ql-v^1)^{n_1+n_2+1}}+\text{lower order terms},\quad n_1,n_2\geq 0,
\end{equation}
where the constant $A_{n_1,n_2}$ is given as
\begin{equation}
	\label{au}
A_{n_1,n_2} = -\frac{(2n_1-1)!!(2n_2-1)!!}{2^{n_1+n_2}}.
\end{equation}
Therefore, it is easy to prove by induction on $g$ that
	\[
		F_g^{KdV}\in C^{\infty}(v) \left[\frac{1}{v^{1,1}}\right][v^{1,1},\dots,v^{1,m_g}],\quad g\geq 2,
		\] 
	for some $m_g\geq 1$. To find $m_g$, we notice that the left-hand side of \eqref{aq} is of the form
	\begin{equation}\label{eq:LHS-loop}
		\frac{1}{v^1-\ql}\diff{F_g^{KdV}}{v^1}+\frac{1}{(v^1-\ql)^2}\sum_{k\geq 1}\frac{k+2}{2}v^{1,k}\diff{F_g^{KdV}}{v^{1,k}}+\sum_{r\geq 3}\frac{A_r}{(v^1-\ql)^r},
		\end{equation}
	where $A_r$ are some expressions computed from gradients of $F_g^{KdV}$, and the right-hand side is of the form
	\[
		\sum_{r\geq 3}\frac{B_r}{(v^1-\ql)^r},
		\]
	for some expressions $B_r$ computed from $F_1^{KdV},\dots, F^{KdV}_{g-1}$. In particular, coefficients $A_r$ and $B_r$ are independent of $\ql$. Note that the variable $\ql$ in the loop equation is considered indeterminate, therefore, the coefficients of $(v^1-\ql)^{-1}$ and $(v^1-\ql)^{-2}$ in~\eqref{eq:LHS-loop} must vanish, and this implies that
	\begin{align}
		&\diff{F_g^{KdV}}{v^1} = 0,\quad g\geq 1,\\
		\label{as}
		&\sum_{k\geq 1}\frac{k+2}{2}v^{1,k}\diff{F_g^{KdV}}{v^{1,k}} = 0,\quad g\geq 2.
	\end{align}
Note that Eq.~\eqref{as} is equivalent to the homogeneity condition of the trivial cohomological field theory. By combining the two identities above, we conclude that $F_g^{KdV}$ is of the form
\begin{equation}
	\label{ar}
	F_g^{KdV} = \sum_{n\geq 0}\sum_{\mu\in P(g,n)} C_{g;\mu}\frac{v^{1,(\mu)}}{(v^{1,1})^{g+n-1}},\quad 
\end{equation}
where $P(g,n)$ is the set of partition $\mu = (\mu_1,\dots,\mu_\ell)$ of $3g-3+n$ with the constraints
\[
\mu_i\geq 2,\quad \ell(\mu):=\ell = n\geq 1. 
	\]
For $\mu = (\mu_1,\dots,\mu_n)\in P(g,n)$, we denote by
\[
	v^{1,(\mu)} = v^{1,\mu_1}\dots v^{1,\mu_n},
	\]
and $C_{g;\mu}$ are some rational numbers. These numbers can be computed either from solving the loop equation or using the intersection numbers on $\mgn$. For example, it is easy to see that 
\[
	C_{g;(3g-2)} = \int_{\oM_{g,1}}\psi_1^{3g-2}.
	\]
The special form \eqref{ar} of $F_g^{KdV}$ for $g\geq 2$ and the explicit expression for $F_1^{KdV}$ imply that 
\begin{equation}
	\label{ay}
		\frac{\qp^n F_g^{KdV}}{\qp v^{1,k_1}\dots\qp v^{1,k_n}}=0 \quad \text{for}\quad k_1+\dots+k_n\neq 3g-3+n,\quad g\geq 1,\ n\geq 1,
\end{equation}
which is a particular form of the general $(3g-2)$-property \cite{dubrovin2001normal,eguchi1998quantum,getzler2004jet}. 

Finally, let us derive a relation that will be used later. By using the above $(3g-2)$ property and the identity \eqref{at}, it is easy to see that the left-hand side of the loop equation, viewed as a polynomial in $\frac{1}{\ql-v^1}$ has the leading term 
\[
	 \frac{(v^{1,1})^{3g-2}}{(\ql-v^1)^{3g-1}}\left(-(3g-2)!+\sum_{k=1}^{3g-2}\binom{3g-2}{k}A_{k-1,3g-1-k}\right)\diff{F_g^{KdV}}{v^{1,3g-2}},
	\]
where the constants $A_{k-1,r-k+1}$ are defined in \eqref{au}. Similarly, we have \[
	\qp_x^{n}\kk{\frac{v^{1,1}}{(v^1-\ql)^3}} = -\frac{(n+2)!(v^{1,1})^{n+1}}{2(\ql-v^1)^{n+3}}+\text{lower order terms},\quad n\geq 0.
	\]
Therefore, when $g\geq 2$, the leading term of the right-hang side is 
\begin{align*}
	&\frac{1}{2}\frac{(v^{1,1})^{3g-2}}{(\ql-v^1)^{3g-1}}\sum_{m=1}^{g-1}A_{3m-1,3g-3m-1}\diff{F_{m}^{KdV}}{v^{1,3m-2}}\diff{F_{g-m}^{KdV}}{v^{1,3g-3m-2}}\\
	+&\frac{1}{2}\frac{(v^{1,1})^{3g-2}}{(\ql-v^1)^{3g-1}}\sum_{\substack{k,\ell\geq 1\\ k+\ell = 3g-4}}A_{k+1,\ell+1}\frac{\qp^2 F_{g-1}^{KdV}}{\qp v^{1,k}\qp v^{1,\ell}}-\frac{1}{16} \frac{(3g-2)! (v^{1,1})^{3g-3}}{(\ql-v^1)^{3g-1}}\diff{F_{g-1}^{KdV}}{v^{1,3g-5}}.
\end{align*}
By comparing the leading terms of both sides of the loop equation, we arrive at the relation
\begin{align}
	\label{av}
	&\left(-(3g-2)!+\sum_{k=1}^{3g-2}\binom{3g-2}{k}A_{k-1,3g-1-k}\right)\diff{F_g^{KdV}}{v^{1,3g-2}}\\
	\notag=&\frac{1}{2}\sum_{m=1}^{g-1}A_{3m-1,3g-3m-1}\diff{F_{m}^{KdV}}{v^{1,3m-2}}\diff{F_{g-m}^{KdV}}{v^{1,3g-3m-2}}\\
	\notag
	+&\frac{1}{2}\sum_{\substack{k,\ell\geq 1\\ k+\ell = 3g-4}}A_{k+1,\ell+1}\frac{\qp^2 F_{g-1}^{KdV}}{\qp v^{1,k}\qp v^{1,\ell}}-\frac{1}{16} \frac{(3g-2)! }{v^{1,1}}\diff{F_{g-1}^{KdV}}{v^{1,3g-5}},\quad g\geq 2.
\end{align}
This relation will play an important role in the proof of Theorem \ref{ae}.

\subsection{Decomposition of free energy functions}
We continue to use the same notations as in previous subsections. Let us first show that, if the genus $g$ free energy function $\mathcal F_g(t^{i,k})$ of $M$ admits a decomposition in terms of the canonical coordinates
\begin{equation}
	\label{am}
	\mathcal F_g(t^{i,k}) = F_g(u^{i,k}) = \sum_{i}\varphi_{g;i}(u) F_g^{KdV}(u^{i,1},\dots,u^{i,3g-2})+H_g,\quad g\geq 1,
\end{equation}
with the function $H_g$ satisfying the condition
\begin{equation}	\label{eq:Property-H-g}
\frac{\qp^n H_g}{\qp u^{i_1,k_1}\dots\qp u^{i_n,k_n}} = 0,\quad k_1+\dots+k_n\geq 3g-3+n,\quad n\geq 1,
\end{equation}
then such decomposition is preserved under both upper triangular and lower triangular infinitesimal symplectic transformation. 

First, we prove the invariance of decomposition \eqref{am} under the upper triangular transformations given by \eqref{an} for an arbitrary upper triangular element $\mathfrak r$. 
\begin{proposition}
	\label{prop:r-action}
	The upper triangular transformations preserve the decomposition \eqref{am} for $g\geq 1$.
\end{proposition}

\begin{proof}
We prove in Lemma~\ref{bb} below that 
\begin{equation}
	\label{ax}
	\left.\frac{\qp^n}{\qp t^{i_1,k_1}\dots\qp t^{i_n,k_n}}\hat{\mathfrak r}[t].\mathcal F_g\right|_{t = 0} = 0,\quad k_1+\dots+k_n\geq 3g-3+n,\quad g,n\geq 1.
\end{equation}
This identity implies that for $k_1+\dots+k_n = 3g-3+n$, the expressions
\[
	\frac{\qp^n F_g}{\qp u^{i_1,k_1}\dots\qp u^{i_n,k_n}}
	\]
remain unchanged after applying the upper triangular transformation. Thus we prove the invariance of the decomposition from the property \eqref{ay} of $F_g^{KdV}$ and decomposition \eqref{am} taken as an assumption at the point where we apply the infinitesimal symplectic transformation.    
\end{proof}

The above proof uses the following lemma:

\begin{lemma}
	\label{bb}
	The identity \eqref{ax} holds true for any upper triangular transformation.
\end{lemma}
\begin{proof}
	First let us recall that the free energy function $\mathcal F_g$ of a cohomological field theory satisfies the condition
	\begin{equation}
		\label{az}
		\frac{\qp^n \mathcal F_g}{\qp t^{j_1,\ell_1}\dots\qp t^{j_n,\ell_n}} = 0,\quad \ell_1+\dots+\ell_n>3g-3+n,\quad g\geq 0,\ n\geq 1.
	\end{equation}
	In what follows, we fix indices $i_1,\dots,i_n$ and $k_1,\dots,k_n$ with
	\begin{equation}
		\label{ba}
		k_1+\dots+k_n = 3g-3+n.
	\end{equation}
	By a straightforward computation using \eqref{an} (see also \cite{lee2009invariance}), we find that the left-hand side of \eqref{ax} reads
	\begin{align*}
		\text{L.H.S.}=& \left.-\sum_{k\geq 1}\sum_{i=1}^N (\mathfrak r_k)^i_{\un}\frac{\qp^{n+1}\mathcal F_g }{\qp t^{i_1,k_1}\dots\qp t^{i_n,k_n}\qp t^{i,k+1}}\right|_{t = 0}\\
		&\left.+\sum_{k\geq 1}\sum_{i,j=1}^N\sum_{m=1}^n (\mathfrak r_k)^i_{i_m}\frac{\qp^{n}\mathcal F_g }{\qp t^{i_1,k_1}\dots\qp t^{i_{m-1},k_{m-1}}\qp t^{i_{m+1},k_{m+1}}\dots\qp t^{i_n,k_n}\qp t^{i,k+k_m}}\right|_{t = 0}\\
		&\left.+\frac{1}{2}\sum_{k,\ell\geq 0}\sum_{i,j=1}^N (-1)^{k+1} (\mathfrak r_{k+\ell+1})^{i,j}\frac{\qp^{n+2}\mathcal F_{g-1}}{\qp t^{i_1,k_1}\dots\qp t^{i_n,k_n}\qp t^{i,k}\qp t^{j,\ell}}\right|_{t = 0}\\
		&\left.+\frac{1}{2}\sum_{k,\ell\geq 0}\sum_{i,j=1}^N\sum_{\substack{g_1+g_2 = g\\I\sqcup J = \{1,\dots,n\}}}(-1)^{k+1} (\mathfrak r_{k+\ell+1})^{i,j}\kk{\diff{}{t^I}\diff{\mathcal F_{g_1}}{t^{i,k}}}\kk{\diff{}{t^J}\diff{\mathcal F_{g_2}}{t^{j,\ell}}}\right|_{t = 0},
	\end{align*}
	here in the last line we use the notation
	\[
		\diff{}{t^I}:=\prod_{m\in I}\diff{}{t^{i_m,k_m}}.
		\]
	By using the identity \eqref{az}, the first line of the right-hand side vanishes due to the fact that
	\[
		k_1+\dots+k_n+k+1\geq 3g-3+n+2>3g-3+(n+1).
		\]
	For a similar reason, the second line and the third line also vanish. As for the last line, we see that it is non-vanishing only when
	\[
		k+\sum_{m\in I} k_m\leq 3g_1-3+(|I|+1),\quad \ell+\sum_{m\in J} k_m\leq 3g_1-3+(|J|+1)
		\]
	which implies that
	\[
		k_1+\dots+k_n\leq 3g-4+n.
		\]
	This contradicts to the assumption \eqref{ba} and hence the last line also vanishes. The proposition is proved.
\end{proof}

Next we prove the invariance of decomposition \eqref{am} under lower triangular transformation. 
\begin{proposition}
	\label{bc}
	The lower triangular transformations preserve the decomposition \eqref{am} for $g\geq 1$.
\end{proposition}
\begin{proof}
	We prove by computing the infinitesimal action of a lower triangular transformation $\mathfrak s$ in terms of  the flat coordinates $v^{\qa}$ of $M$. Such an action is given in \cite{buryak2012polynomial} and reads 
	\[
		\hat{\mathfrak s}[v].F_g =\hat{\mathfrak s}[t].\mathcal F_g-\sum_{s\geq 0}\sum_{j=1}^N \diff{F_g}{v^{j,s}}\kk{\diff{}{t^{\un,0}}}^{s+1}\diff{}{t^{j,0}}\hat{\mathfrak s}[t].\mathcal F_0,
		\]
        here we denote by $\hat{\mathfrak s}[v].$ the action of $\mathfrak s$ in terms of flat coordinates.
        For general lower triangular transformation given by 
\begin{align*}
	\hat{\mathfrak s} =& -\frac{1}{2\qe^2}(\mathfrak s_3)_{\un,\un}+\frac{1}{\qe^2}\sum_{k\geq 0}\sum_{i=1}^N (\mathfrak s_{k+2})_{\un,i}t^{i,k}+\frac{1}{2\qe^2}\sum_{k,\ell\geq 0}\sum_{i,j=1}^N (-1)^i(\mathfrak s_{k+\ell+1})_{i,j}t^{i,k}t^{j,\ell}+\mathcal D,\\
	\mathcal D =&-\sum_{i=1}^N(\mathfrak s_1)^i_{\un}\diff{}{t^{i,0}}+\sum_{k\geq 0,\ell\geq 1}\sum_{i,j=1}^N(\mathfrak s_\ell)^i_j t^{j,k+\ell}\diff{}{t^{i,k}},
\end{align*}
it is easy to see by definition that 
\begin{align*}
\hat{\mathfrak s}[t].\mathcal F_g =&\,\mathcal D(\mathcal F_g),\quad g\geq 1,\\
\hat{\mathfrak s}[t].\mathcal F_0 =&\, -\frac{1}{2}(\mathfrak s_3)_{\un,\un}+\sum_{k\geq 0}\sum_{i=1}^N (\mathfrak s_{k+2})_{\un,i}t^{i,k}+\frac{1}{2}\sum_{k,\ell\geq 0}\sum_{i,j=1}^N (-1)^i(\mathfrak s_{k+\ell+1})_{i,j}t^{i,k}t^{j,\ell}+\mathcal D(\mathcal F_0).
\end{align*}
For $g\geq 1$, it follows from the fact $\mathcal F_g = F_g(v^{j,s})$ that 
\begin{align*}
\hat{\mathfrak s}[t].\mathcal F_g =&\, \sum_{s\geq 0}\sum_{j=1}^N\diff{F_g}{v^{j,s}}\kk{\diff{}{t^{\un,0}}}^{s+1}\diff{}{t^{j,0}}\mathcal D(\mathcal F_0),
\end{align*}
here we use the relation 
\[
v^{j,s} = \kk{\diff{}{t^{\un,0}}}^{s+1}\diff{\mathcal F_0}{t^{j,0}}.
\]
Therefore, it is immediate to arrive at the action 
	\begin{equation}\label{eq:s-translation}
		\hat{\mathfrak s}[v].F_g = -\frac{1}{2}\sum_{i=1}^N (-1)^i\kk{(\mathfrak s_{1})_{i,\un}+(\mathfrak s_{1})_{\un,i}}\diff{F_g}{v^i}.
		\end{equation}
	This identity means that the lower triangular transformations act as infinitesimal translations of the flat coordinates. This leaves the decomposition \eqref{am} invariant. The proposition is proved.
\end{proof}

Combining the result or Proposition~\ref{prop:r-action} and~\ref{bc}, we can conclude that 
\begin{proposition} \label{prop:decomposition-Givental}
Decomposition~\eqref{am} holds true for any partition function given by the Givental formula.    
\end{proposition}
\begin{proof}
 Indeed, decomposition~\eqref{am} holds for the partition function given by the $N$ copies of the (possibly rescaled) KdV partition functions. On the other hand, any partition function given by the Givental formula is obtained by a combination of an upper-triangular and a lower-triangular symplectic tranformations applied to the $N$ copies of the (possibly rescaled) KdV partition functions. 

 Proposition~\ref{prop:r-action} and~\ref{bc} imply that the infinitesimal actions of the upper-trangular and lower-triangular symplectic transformations preserves the decomposition~\eqref{am}. In other words, the involved infinitesimal transformations act nontrivially only on the coefficients $\varphi_{g;i}$ and $H_g$, preserving the property~\eqref{eq:Property-H-g} of $H_g$. 
 
 Integrating these infinitesimal Lie algebra actions to the actions of the corresponding Lie group elements, we obtain the desired decomposition property for any partition function that we can reach using the upper-triangula and lower-triangular symplectic transformation starting from the $N$ copies of the (possibly rescaled) KdV partition functions. In other words, we obtain the desired decomposition property for any partition function given by the Givental formula.    
\end{proof}

\begin{remark}
Note that a possible alternative to the infinitesimal analysis of Proposition~\ref{prop:r-action} would be to use the closed graphical formula for the upper traingular Givental group action~\cite{dunin2013givental}, which would lead to the same result.    
\end{remark}

Finally, we can prove Theorem \ref{ae}. For this, we need to use the loop equation \eqref{ao} of $M$, hence let us first recall some basic facts about semisimple Frobenius manifolds. One may refer to \cite{dubrovin1996geometry,dubrovin1999painleve,dubrovin2001normal} for details. As before, we use $(v^\qa)$ to denote the flat coordinates of $M$ and $(u^i)$ to denote its canonical coordinates. Recall that in terms of $(u^i)$, the flat metric $\eta$ of $M$ is diagonal of the form $\eta = \sum f_i (du^i)^2$. We define functions $\Psi_{i\qa}$ by
\[
\diff{u^i}{v^\qa} = \frac{\Psi_{i\qa}}{\Psi_{i1}},\quad \Psi_{i1} = \sqrt{f_i}.
	\]
These functions satisfy the identity
\begin{equation}
	\label{bd}
	\Psi_{i\qa}\Psi_j^\qa =\qd_{ij},\quad  \Psi_j^\qa:=\eta^{\qa\qb}\Psi_{j\qb}.
\end{equation}
For $i\neq j$, we denote by $\qg_{ij}$ the rotation coefficients of the flat metric $\eta$, and in terms of the canonical coordinates, they read
\[
	\qg_{ij} = \frac{1}{2\sqrt{f_if_j}}\diff{f_i}{u^j},\quad i\neq j.
	\]
Define functions $V_{ij}$ to be
\[
	V_{ij} = \begin{cases}
		0, & \text{for } i=j,\\
        (u^j-u^i)\qg_{ij}, & \text{for } i\neq j.
	\end{cases}
	\]
{Note that the functions $\qg_{ij}$ are symmetric with respect to their indices and hence  $V_{ij}$ are antisymmetric. }By using the notations above, the Gauss-Manin system satisfied by periods $p_\qa(u;\ql)$ of $M$ can be written into the following first-order equations for functions $\phi_{i}$:
\begin{align*}
	\diff{\phi_{i}}{u^j}&=-\frac{V_{ij}}{u^i-u^j}\phi_{j},\quad i\neq j,\\
	\diff{\phi_{i}}{u^i}&=\frac{1}{\ql-u^i}\kk{\frac12 \phi_{i}+\sum_{j} V_{ij}\phi_{j}}+\sum_j\frac{V_{ij}}{u^i-u^j}\phi_{j},\\
	\diff{\phi_{i}}{\ql} &= \frac{\phi_i}{2(u^i-\ql)}+\sum_j\frac{V_{ij}}{u^i-\ql}\phi_j.
\end{align*}
Let us denote by $\phi_{i\qa}$ a fundamental solution matrix for the above system, then the periods $p_\qa$ are specified by
\[
	\diff{p_\qa}{u^i} = \Psi_{i1}\phi_{i\qa},\quad \diff{p_\qa}{\ql} = -\sum_i \Psi_{i1}\phi_{i\qa}.
	\]
By definition, functions $p_\qa$ serve as flat coordinates of the flat pencil of $M$, and we denote by $G^{\qa\qb}$ the corresponding Gram matrix. In particular, we have
\[
	\phi_{i\qa}G^{\qa\qb}\phi_{j\qb} = \frac{\qd_{ij}}{u^i-\ql}.
	\]
The following proposition is important in analyzing the structure of the loop equation \eqref{ao}.
\begin{proposition}
	\label{bf}
	We have:
	\begin{enumerate}
		\item In terms of canonical coordinates, the function \[
			\qp_x^{n_1}(\qp^\qg p_\qa) G^{\qa\qb}\qp_x^{n_2}(\qp^\rho p_\qb)\in \mathcal A\left[\frac{1}{\ql-u^1},\dots,\frac{1}{\ql-u^N}\right],
			\]
		and, viewed as a polynomial in $\frac{1}{\ql-u^1},\dots,\frac{1}{\ql-u^N}$, it has the leading term
		\[
			\sum_i A_{n_1,n_2}\frac{(u^{i,1})^{n_1+n_2}}{(\ql-u^i)^{n_1+n_2+1}}\Psi_i^\qg\Psi_i^\rho,
			\]
		here the constants $A_{n_1,n_2}$ are defined in \eqref{au}.
		\item The function
		\[\qp_x^{n}\left[\nabla\diff{p_\qa}{\ql}\cdot \nabla\diff{p_\qb}{\ql}\cdot v_x\right]^\qg G^{\qa\qb}\in \mathcal A\left[\frac{1}{\ql-u^1},\dots,\frac{1}{\ql-u^N}\right],
			\]
		and, viewed as a polynomial in $\frac{1}{\ql-u^1},\dots,\frac{1}{\ql-u^N}$, it has the leading term
		\[
			-\sum_i \frac{(n+2)!(u^{i,1})^{n+1}}{8(\ql-u^i)^{n+3}}\frac{\Psi_i^\qg}{\Psi_{i1}}.
			\]
	\end{enumerate}
\end{proposition}
\begin{proof}
	The first statement is proved in Lemma 3.10.19 of \cite{dubrovin2001normal}. As for the second statement, by a straightforward computation using \eqref{be} and the Gauss-Manin system satisfied by $p_\qa$, we see that, in terms of canonical coordinates,
		\[\left[\nabla\diff{p_\qa}{\ql}\cdot \nabla\diff{p_\qb}{\ql}\cdot v_x\right]^\qg G^{\qa\qb} = \sum_i\frac{u^{i,1}\Psi_i^\qg}{\Psi_{i1}}\kk{\frac{1}{4(u^i-\ql)^3}+\frac{1}{(u^i-\ql)^2}\sum_j\frac{V_{ij}^2}{u^j-\ql}},
		\]
	here we also use \eqref{bd} and the fact \cite{dubrovin1996geometry} that 
	\begin{equation}
		\label{bi}
		c_{\qa\qb}^\qg = \sum_i\frac{\Psi_{i\qa}\Psi_{i\qb}\Psi_i^\qg}{\Psi_{i1}}.
	\end{equation}
	Then the second statement can be verified directly. The proposition is proved. 
\end{proof}
\begin{proposition}
	\label{prop:functions-varphi}
	The function $\varphi_{g;i}$ defined in \eqref{am} is given by $(f_i(u))^{1-g}$.
\end{proposition}
\begin{proof}
	By using Dubrovin-Zhang's theorem that their tau-function is given by the Givental formula~\cite[Theorem 4.3.14]{dubrovin2005normal}, (or, alternatively, one can employ Teleman's result \cite{teleman2012structure} on the classification of cohomological field theories, and identify the Virasoro constraints constructed by Givental \cite{givental2001gromov} and those constructed by Dubrovin and Zhang \cite{dubrovin2001normal}), we know that $F_g$ given by the Givental formula satisfies the loop equation \eqref{ao}. 
	
	Let us prove $\varphi_{g;i} = (f_i(u))^{1-g}$ by induction on $g$. For $g = 1$, it is well-known that 
	\[
		F_1 = \frac{1}{24}\sum_i \log u^{i,1}+H_1(u),
		\] 
	hence we see that $\varphi_{1;i} = 1$. Assume that we have proved the statement for $1,\dots,g-1$, let us find $\varphi_{g;i}$. It follows from Lemma 3.10.19 of \cite{dubrovin2001normal} that the left-hand side of \eqref{ao}, viewed as a polynomial in $\frac{1}{\ql-u^1},\dots,\frac{1}{\ql-u^N}$, has the leading term
	\[
		\sum_i \frac{(u^{i,1})^{3g-2}}{(\ql-u^i)^{3g-1}}\left(-(3g-2)!+\sum_{k=1}^{3g-2}\binom{3g-2}{k}A_{k-1,3g-1-k}\right)\diff{F_g}{u^{i,3g-2}}.
		\]
	Similarly, it follows from Proposition~\ref{bf} that the right-hand side, viewed as a polynomial in $\frac{1}{\ql-u^1},\dots,\frac{1}{\ql-u^N}$, has the leading term
		\begin{align*}
			&\frac{1}{2}\sum_i\frac{(u^{i,1})^{3g-2}}{(\ql-u^i)^{3g-1}}\sum_{m=1}^{g-1}A_{3m-1,3g-3m-1}\frac{1}{\Psi_{i1}^2}\diff{F_{m}}{u^{i,3m-2}}\diff{F_{g-m}}{u^{i,3g-3m-2}}\\
			+&\frac{1}{2}\sum_i\frac{(u^{i,1})^{3g-2}}{(\ql-u^i)^{3g-1}}\sum_{\substack{k,\ell\geq 1\\ k+\ell = 3g-4}}A_{k+1,\ell+1}\frac{1}{\Psi_{i1}^2}\frac{\qp^2 F_{g-1}}{\qp u^{i,k}\qp u^{i,\ell}}-\frac{1}{16}\sum_i\frac{1}{\Psi_{i1}^2} \frac{(3g-2)! (u^{i,1})^{3g-3}}{(\ql-u^i)^{3g-1}}\diff{F_{g-1}}{u^{i,3g-5}}.
		\end{align*}
Now by using \eqref{am}, we have for $g\geq 1$
\begin{align*}
	\diff{F_g}{u^{i,3g-2}} &= \varphi_{g;i}\diff{F_g^{KdV}(u^i)}{u^{i,3g-2}},\\
	\frac{\qp^2F_g}{\qp u^{i,k}\qp u^{j,\ell}} &= \qd_{ij}\varphi_{g;i}\frac{\qp^2F_g^{KdV}(u^i)}{\qp u^{i,k}\qp u^{i,\ell}},\quad k+\ell = 3g-1.
\end{align*}
Using above identities, as well as the recursion relation \eqref{av} and the induction hypothesis $\varphi_{m;i} = (f_i(u))^{1-m}$ for $m\leq g-1$, we have
\[
\varphi_{g;i} = \frac{(f_i(u))^{2-g}}{\Psi_{i1}^2} =(f_i(u))^{1-g}. 
	\]
The proposition is proved.
\end{proof}

Now we are fully armed to present the proof of Theorem~\ref{ae}. 

\begin{proof}[Proof of Theorem~\ref{ae}]

First, we use Proposition~\ref{prop:decomposition-Givental}, which states that we indeed 
have the general shape of the decomposition~\eqref{ah} --- it is given in Equation~\eqref{am}, with the unknown functions $\varphi_{g;i}$ and $H_g$. Note, however, that Proposition~\ref{prop:decomposition-Givental} implies property~\eqref{eq:Property-H-g} demanded for $H_g$ in~\eqref{eq:RestrictionOnH}. 

Thus, the only missing bit to pass from Equation~\eqref{am} to Equation~\eqref{ah} is the explicit computation of the coefficients $\varphi_{g;i}$. This is done in Proposition~\ref{prop:functions-varphi}, where we show that 
$\varphi_{g;i}(u) = f_i(u)^{1-g}$. The theorem is proved. 
\end{proof}

	\section{Universal identities of free energy functions}
	The goal of this section is to present a set of operators $O_{\{\alpha_1,a_1;\dots,\alpha_n,a_n\}}$ of order $n$ in the time variables that are going to be used to derive the universal relations. The parameters here are subject to the following condition: $a_i \geq 2$,  $i=1,\dots,n$. 
The main reason to introduce these operators is the following list of properties:
\begin{theorem} \label{thm:properties-of-operators-O}
We have:
	\begin{enumerate}
		\item The operators $O_{\{\alpha_1,a_1;\dots,\alpha_n,a_n\}}$ commute with $v^{\beta,0}\cdot $ and $v^{\beta,1}\cdot$ (that is, with the operators of multiplication by $v^{\beta,0}$ and $v^{\beta,1}$). 
		\item Let $m \geq 1$; $b_i\geq 2$, $i=1,\dots,m$, and 
		$\sum_{i=1}^m b_i \leq \sum_{i=1}^n a_i -1$.
		Then 
		\begin{align}
			O_{\{\alpha_1,a_1;\dots,\alpha_n,a_n\}} \big( v^{\beta_1,b_1}\cdots v^{\beta_m,b_m} \big) = 0.
		\end{align}  
		\item Let $m\geq 1$; $b_i\geq 2$, $i=1,\dots,m$, and $\sum_{i=1}^m b_i=\sum_{i=1}^n a_i$. Then if $m\geq n$, we have
		\begin{align}
			O_{\{\alpha_1,a_1;\dots,\alpha_n,a_n\}} \big( v^{\beta_1,b_1}\cdots v^{\beta_m,b_m} \big) = \delta_{nm} \sum_{\sigma\in S_n} \prod_{i=1}^n \delta_{a_ib_{\sigma(i)}} \prod_{i=1}^n \prod_{j=0}^{a_i} \partial_x 
			\ddd{\qt^{\gamma_{i,j}}_0\qt_{\gamma_{i,j+1},0}},
		\end{align}  
		where $\gamma_{i,0} = \beta_{\sigma(i)}$ and $\gamma_{i,a_i} = \alpha_i$. 
	\end{enumerate}
\end{theorem}

Below we introduce the construction of these operators and, after we discuss their properties, we prove all statements of Theorem~\ref{thm:properties-of-operators-O} in Corollaries~\ref{cor:Thm-Ooper-m1},~\ref{cor:OperO-commute01}, and~\ref{cor:OperO-mgeq2}. As an application, we derive some universal identities and prove Theorem~\ref{af}.
	\label{ad}
	
\subsection{A useful set of operators} 
	\label{bq}
We start by explaining the notations used for defining the operators $O_{\{\alpha_1,a_1;\dots,\alpha_n,a_n\}}$.
	\subsubsection{Basic notation for trees}
	Let $RT_{n}$ be the set of stable rooted trees with $n$ legs $\sigma_1,\dots,\sigma_n$. We demand that the index of each vertex except for the root is at least $3$. For a $T\in RT_{n}$ we use the following notation:
	\begin{itemize}
		\item $H(T)$ is the set of half-edges of $T$.
		\item $L(T)$ is the set of  legs of $T$.
		\item $H_e(T)\coloneqq H(T)\setminus L(T)$.
		\item $\iota\colon H_e(T)\to H_e(T)$ is the involution that interchanges the half-edges that form an edge.
		\item $E(T)$ is the set of edges of $T$, $E\cong H_e(T)/\iota$.
		\item $H_+(T)\subset H(T)$ is the set of the so-called ``positive'' half-edges that consists of all  legs of $T$ and of half-edges in $H(T)\setminus L(T)$ directed away from the root at the vertices where they are attached,
		$H_+(T)\cong E(T)\cup L(T)$; 
		\item $H_-(T)\subset H(T)$ is the set of the so-called ``negative'' half-edges that consists of all half-edges in $H(T)\setminus L(T)$ directed towards the root at the vertices where they are attached, $H_-(T)\cong E(T)$;
		\item $V(T),V_{nr}(T)$ are the sets of vertices and non-root vertices of $T$. 
		\item $v_r\in V(T)$ is the root vertex of $T$; $V(T)=\{v_r(T)\}\sqcup V_{nr}(T)$.
		\item For a $v\in V(T)$, $H_+(v)$ is the set of all positive half-edges attached to $v$. 
		\item For a $v\in V_{nr}(T)$ let $H_-(v)$ be the negative half-edge attached to $v$.
		\item We say that a vertex or a (half-)edge $x$ is a descendant of a vertex or a (half-)edge $y$ if $y$ is on the unique path connecting $x$ to $v_r$. 
		\item For an $h\in H_+(T)$ let $DL(h)$ be the set of all legs that are descendants to $h$, including $h$ itself. Note that $DL(h)\subseteq L(T)$ for any $h\in H_+(T)$ and $DL(l)=\{l\}$ for $l\in L(T)$. 
		\item For an $h\in H_+(T)$ let $DH(h)$ be the set of all positive half-edges that are descendants to $h$, \emph{excluding} $h$. For instance, for $l\in L(T)$ we have $DH(l) = \emptyset$, and for $h\in H_+(T)\setminus L(T)$ we have $DH(h) \supseteq DL(h)$. 
	\end{itemize}
	
	In the pictures it is convenient to arrange the half-edges at each vertex such that the negative half-edges are directed to the right and the positive half-edges are directed to the left (this convention is opposite to the one used in the similar structures in~\cite{buryak2022conjectural,buryak2022tautological}, but it is more suitable for the purpose of defining the differential operators).
	 In particular, the root vertex is the rightmost vertex on the pictures. Here is an example of a stable rooted tree in $\RT_{4}$ placed on the plane following this convention:
	\begin{gather*} 
		\vcenter
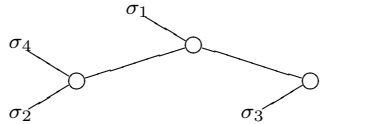
{\xymatrix@C=15pt@R=5pt{
				& & & & & & \\
				& & & *+[o][F-]{{}} \ar@{-}[lu]*{{}_{\,\,\sigma_1}} \ar@{-}[lld] & & & \\ 
				& *+[o][F-]{{}}  \ar@{-}[ld]*{{}_{\,\,\sigma_2}} \ar@{-}[lu]*{{}_{\,\,\sigma_4}} & &  & & *+[o][F-]{{}} \ar@{-}[llu]\ar@{-}[ld]*{{}_{\,\sigma_3}}  & \\
				& & & & & &}}\,.
	\end{gather*}
	
	Let $T\in RT_n$. Introduce an extra function $q\colon H_+(T)\to \ZZ_{\geq 0}$ such that 
	\begin{align}
	\sum_{h\in H_+(T)} q(h) + |E(T)| =\sum_{i=1}^n a_i
	\end{align}
	and at each vertex $v\in V_{nr}(T)$ we have 
	\begin{align}
		|H_+(v)|-2 \geq \sum_{h\in H_+(v)} q(h). 
	\end{align}
	Let $Q(T,\sum_{i=1}^n a_i)$ denote the set of such functions. We also associate to each $h\in H(T)$ an index $\alpha(h)$ such that $\alpha(\sigma_i) = \alpha_i$.

	\subsubsection{Eguchi-Xiong operators and their generalizations}
Let $O_{\beta,p}$ be the differential operators of Eguchi-Xiong~\cite{eguchi1998quantum} (see also~\cite{buryak2012polynomial,liu2002quantum}) defined inductively by the following rules:
\begin{align}
	\label{bo}
	O_{\beta,0}& \coloneqq \frac{\partial}{\partial t^{\beta,0}} ;\\ 
	\label{bp}
	O_{\beta,p}& \coloneqq \frac{\partial}{\partial t^{\beta,p}} - \sum_{k=0}^{p-1}  \ddd{\qt^{\gamma}_0\qt_{\beta,k}} O_{\qg,p-k-1}.
\end{align}
We have $[O_{\beta,p},O_{\gamma,q}]=0$ if $p,q\geq 1$. 

We associate to a pair $(T,q)$, $T\in RT_n$, $q\in Q(T,\sum_{i=1}^n a_i )$, an operator $O_{(T,q)}$ given by 
\begin{align}
	O_{(T,q)}  \coloneqq & (-1)^{|E(T)|} \prod_{h\in H_-(T)} \eta^{\alpha(h),\alpha(\iota(h))}
	\\
	\notag
	&  
	\frac{1}{\prod_{i=1}^n (a_i+1)!} \prod_{h\in H_+(T)} \Big(\sum_{\ell\in DL(h)}(a(\ell)+1) - \sum_{h'\in DH(h)} (q(h')+1) \Big)_{q(h)+1}
	\\
	\notag
	& 
	\Bigg( \prod_{v\in V_{nr}(T)} 
	\bigg( \nordbullet\prod_{h\in H_+(v)} O_{\alpha(h),q(h)}\nordbullet\, 
	\ddd{\qt_{\alpha(h_-(v)),0}} \bigg) \Bigg) 
	\nordbullet\prod_{h\in H_+(v_r)} O_{\alpha(h),q(h)}\nordbullet
\end{align}
	{Here $a(\ell)=a_i$ for $\ell = \qs_i$ }and $(s)_{t}=s(s-1)\cdots (s-t+1)$ denotes the Pochhammer symbol and we recall that it is assumed that $\sum_{h\in H_+(T)} q(h) + |E(T)| =  \sum_{i=1}^n a_i$. By normal order we just mean that we put differentiations in the vector fields ahead, that is, the vector fields are not allowed to act on the coefficients of each other. 
	
\begin{remark} In the notation of Liu~\cite{liu2002quantum}, the function 
\[\nordbullet\prod_{h\in H_+(v)} O_{\alpha(h),q(h)}\nordbullet\, 
	\ddd{\qt_{\alpha(h_-(v)),0}}\]
can be expressed in the form 
\[
\eta_{\alpha(h_-(v))\alpha(\iota(h_-(v)))}\ddd{ T^0(e^{\alpha(\iota(h_-(v)))}) \prod_{h\in H_+(v)} T^{q(h)}(e_{\alpha(h)}) }.
\]
\end{remark}

\subsubsection{Definition of operators}

\begin{definition}
The operator $O_{\{\alpha_1,a_1;\dots,\alpha_n,a_n\}}$ is defined as
\begin{align}
	O_{\{\alpha_1,a_1;\dots,\alpha_n,a_n\}} \coloneqq \sum_{T\in RT_n} \sum_{q\in Q(T,\sum_{i=1}^n a_i )} O_{(T,q)}
\end{align}
\end{definition}

For instance, 
\begin{align}
        \notag
	O_\emptyset & = \Id; \\ 
        \notag 
	O_{\{\alpha_1,a_1\}} & = O_{\alpha_1,a_1}; 
	\\ 
        \label{aw}
	O_{\{\alpha_1,a_1;\alpha_2,a_2\}} & = \nordbullet O_{\alpha_1,a_1} O_{\alpha_2,a_2} \nordbullet 
	-  \frac{(a_1+a_2)!}{a_1! a_2!} \langle \! \langle  T^{0}(e_{\alpha_1})  T^{0}(e_{\alpha_2}) T^0(e^{\alpha}) \rangle \! \rangle_0 O_{\alpha,a_1+a_2-1}.
\end{align}
\subsection{Action on $v^{\beta,b}$} Consider $O_{\{\alpha_1,a_1;\dots;\alpha_n,a_n\}}v^{\beta,b}$ for $b\leq \sum_{i=1}^n a_i -1$. Note that the usage of the Eugchi-Xiong operators implies that we are dealing with a lift of a tautological relation from the moduli space of curves $\oM_{0,m+b+2}$~\cite{liu2002quantum}. Let us describe this relation. 

For each $T\in RT_n$ let $T_{\sim}$ be the tree with $b+2$ extra labeled legs attached to the root vertex $v_r$. This graph defines a stratum in $\oM_{0,n+b+2}$ and let $\gl_{T_\sim}$ be the boundary map corresponding to this stratum. To fix the notation we choose an order on the legs attached to each vertex $v$ using the map $o_v\colon H_+(v) \to \{1,\dots,|H_+(v)|\}$.

For each $q\in Q(T,\sum_{i=1}^n a_i )$ consider class $B_{(T,q)}$ given by 
\begin{align}
	B_{(T,q)}& \coloneqq \frac{(-1)^{|E(T)|}}{\prod_{i=1}^n (a_i+1)!} \prod_{h\in H_+(T)} \Big(\sum_{l\in DL(h)}(a(l)+1) - \sum_{h'\in DH(h)} (q(h')+1) \Big)_{q(h)+1}
	\\ \notag 
	& \quad (\gl_{T_{\sim}})_* \bigg[ \Big(\prod_{h\in H_+(v_{r})} \psi_{o_{v_r}(h)}^{q(h)}\Big)\Big|_{\oM_{0,|H_+({v_r})|+\chi+2}} \otimes \bigotimes_{v\in V_{nr}(T')} \Big(\prod_{h\in H_+(v)} \psi_{o_v(h)}^{q(h)}\Big)\Big|_{\oM_{0,|H_+(v)|+1}} \bigg]
	\\ \notag
	& \in R^{\sum_{i=1}^n a_i } (\oM_{0,n+b+2}).
\end{align}
Furthermore, define 
\begin{align}
B_{a_1,\dots,a_n}\coloneqq \sum_{T\in RT_n} \sum_{q\in Q(T,\sum_{i=1}^n a_i )} B_{(T,q)}\in R^{\sum_{i=1}^n a_i } (\oM_{0,n+b+2}),
\end{align}
where $R^d(\oM_{0,n+b+2})$, $d\geq 0$, denotes the degree $d$ component of the tautological ring of $\oM_{0,n+b+2}$ (which coincides with the Chow ring in this case).

\begin{lemma}\label{lem:FirstVanishing} For $b\leq \sum_{i=1}^n a_i -1$ we have $B_{a_1,\dots,a_n}=0$.
\end{lemma}

\begin{proof} For any $b$, $a_1,\dots,a_n$, we have $B_{a_1,\dots,a_n}=0$ for $a_1+\cdots+a_n\geq b+1$, see~\cite[Theorem 3.4 in combination with Theorem 2.3]{buryak2022tautological}. 
\end{proof}

\begin{remark} \label{rem:dimension-vanishing}
	Note also that $\dim\oM_{0,n+b+2}=n+b-1$, so the statement of Lemma~\ref{lem:FirstVanishing} holds for dimensional reasons for $\sum_{i=1}^n a_i \geq n+b$, that is, for $b\leq \sum_{i=1}^n a_i -n$. 
\end{remark}

An immediate corollary of Lemma~\ref{lem:FirstVanishing} is

\begin{corollary} \label{cor:Thm-Ooper-m1}  The statements (2) and (3) of Theorem~\ref{thm:properties-of-operators-O} hold for $m=1$. That is, 
\begin{align}
	O_{\{\alpha_1,a_1;\dots;\alpha_n,a_n\}}v^{\beta,b} = 0
\end{align}	
	for $b\leq \sum_{i=1}^n a_i -1$, $n\geq 1$, and 
\begin{align}
O_{\alpha,a}v^{\beta,b} = \prod_{j=0}^{a} \partial_x 	\ddd{\qt^{\gamma_{j}}_0\qt_{\gamma_{j+1},0}}
\end{align}	
 with $\gamma_0=\beta$ and $\gamma_a=\alpha$ for $b=a$. 
\end{corollary}

\begin{proof} The first statement that $O_{\{\alpha_1,a_1;\dots;\alpha_n,a_n\}}v^{\beta,b} = 0$ for $b\leq \sum_{i=1}^n a_i -1$, $n\geq 2$ from Lemma~\ref{lem:FirstVanishing}. Using this lemma and the standard conversion of tautological relations into the PDEs for the descendant potentials, see e.~g.~\cite[Sec.~2.1.3]{faber2010tautological} or~\cite{liu2002quantum}, we obtain the desired vanishing. 

The second statement is merely an exercise on iterative application of the topological recursion relation in genus $0$~\cite{witten1990two}. 
\end{proof}

\begin{remark} Lemma~\ref{lem:FirstVanishing} and Corollary~\ref{cor:Thm-Ooper-m1} explain, in particular, the necessity of so involved definition of the operators $O_{\{\alpha_1,a_1;\dots;\alpha_n,a_n\}}$. For instance, one could try to use simpler operators $\tilde O_{\{\alpha_1,a_1;\dots;\alpha_n,a_n\}} = \nordbullet\prod_{i=1}^n O_{\alpha_i,a_i}\nordbullet$ instead. But then the only vanishing property would be the dimensional vanishing as in Remark~\ref{rem:dimension-vanishing}, which is a much weaker property. A more detailed discussion is given in Remark \ref{revre}.
\end{remark}

\subsection{The iterative structure of operators} Consider  $O_{\{\alpha_1,a_1;\dots;\alpha_n,a_n\}}(f_1f_2)$, where $f_1,f_2$ are some functions in  jet variables $v^{\beta,b}$.

\begin{lemma} \label{lem:FactorizationO} We have 
\begin{align}
	O_{\{\alpha_1,a_1;\dots;\alpha_n,a_n\}}(f_1f_2) = \sum_{\substack{I_1\sqcup I_2 
			= \{1,\dots,n\}}	} O_{\{\alpha_i,a_i\}_{i\in I_1}}(f_1) O_{\{\alpha_i,a_i\}_{i\in I_2}}(f_2) 
\end{align}	
\end{lemma}

\begin{proof} This identity follows directly from the structure of the operators. {Indeed, for $T\in RT_n$ and  $q\in Q(T,\sum_{i=1}^n a_i)$, each factor of the vector fields in the product $\nordbullet\prod_{h\in H_+(v_r)} O_{\alpha(h),q(h)}\nordbullet$ in $O_{(T,p)}$ acts on $f_1f_2$ by the Leibniz rule.} This splits $H_+(v_r)$ into two subsets, $H_+(v_r)=H_+(v_r)_1\sqcup H_+(v_r)_2$ such that vector fields $O_{\alpha(h),q(h)}$ with $h\in H_+(v_r)_i$ are applied to $f_i$, $i=1,2$. 
	
Let $T_1$ and $T_2$, respectively, be the trees obtained by contracting to the root vertex the full subtree descending to $H_+(v_r)_2$ and $H_+(v_r)_1$, respectively. 
Let $\sqcup_{h \in H_+(v_r)_i} DL(h)=\{\sigma_j\}_{j\in I_i}$, $i=1,2$. 
Then $T_i\in RT_{|I_i|}$ (with the legs labeled by $\sigma_j$, $j\in I_i$), and $q_i\coloneqq q|_{H_+(T_i)}\in Q(T_i,\chi_i)$, $i=1,2$, where $\chi_i=(\sum_{j\in I_i} a_j)-1$. This allows to rearrange the sum
\begin{align}
	\sum_{T\in RT_n} \sum_{q\in Q(T,\chi)} O_{(T,q)} (f_1f_2)
\end{align}
as 
\begin{align}
	\sum_{\substack{I_1\sqcup I_2 
			= \{1,\dots,n\}}	} 
	\bigg(\sum_{T_1\in RT_{|I_1|}} \sum_{q_1\in Q(T_1,\chi_1)} O_{(T_1,q_1)} (f_1)\bigg)
	\bigg(\sum_{T_2\in RT_{|I_2|}} \sum_{q_2\in Q(T_2,\chi_2)} O_{(T_2,q_2)} (f_2)\bigg),
\end{align}
which implies the statement of the lemma. 
\end{proof}

An immediate corollary of this lemma is the following statement:
\begin{corollary} \label{cor:LeibnizRule-O-multi}
For any functions $f_1,\dots,f_m$ in jet variables $v^{\beta,b}$ we have:
	\begin{align}
		\label{bk}
		O_{\{\alpha_1,a_1;\dots;\alpha_n,a_n\}}(f_1\cdots f_m) = \sum_{\substack{I_1\sqcup \cdots \sqcup I_m  \\ 
				= \{1,\dots,n\}}	} \prod_{j=1}^m O_{\{\alpha_i,a_i\}_{i\in I_j}}(f_i).
	\end{align}	
\end{corollary}

And also now we can proof the first statement of Theorem~\ref{thm:properties-of-operators-O}:
\begin{corollary} \label{cor:OperO-commute01} The operators $O_{\{\alpha_1,a_1;\dots;\alpha_n,a_n\}}$ commute with the operators of multiplication by $v^{\beta,0}$ and $v^{\beta,1}$.
\end{corollary}
\begin{proof} Indeed, by Lemma~\ref{lem:FirstVanishing} we have 
	\begin{align}
		[O_{\{\alpha_1,a_1;\dots;\alpha_n,a_n\}}, v^{\beta,b}\cdot] = \sum_{\substack{I_1\sqcup I_2 
				= \{1,\dots,n\}\\ I_1\not=\emptyset}	} O_{\{\alpha_i,a_i\}_{i\in I_1}}(v^{\beta,b}) O_{\{\alpha_i,a_i\}_{i\in I_2}}
	\end{align}
Note that $\sum_{i\in I_1} a_i \geq 2n$, hence for $b=0,1$ we have $b\leq a-1$ for $n=1$ and $b\leq \sum_{i\in I_1} a_i+1-n$ for $n\geq 2$. Then by Corollary~\ref{cor:Thm-Ooper-m1}  $O_{\{\alpha_i,a_i\}_{i\in I_1}}(v^{\beta,b})$ vanishes for all $I_1\neq\emptyset$. 
\end{proof}

Now we can use Corollary~\ref{cor:LeibnizRule-O-multi} to prove the statements (2) and (3) of Theorem~\ref{thm:properties-of-operators-O} hold for $m\geq 2$. We have:

\begin{corollary} \label{cor:OperO-mgeq2} Let $m \geq 2$; $b_i\geq 2$, $i=1,\dots,m$, and 
$\sum_{i=1}^m b_i \leq \sum_{i=1}^n a_i -1$.
Then 
\begin{align}
	O_{\{\alpha_1,a_1;\dots,\alpha_n,a_n\}} \big( v^{\beta_1,b_1}\cdots v^{\beta_m,b_m} \big) = 0.
\end{align}  
If $\sum_{i=1}^m b_i=\sum_{i=1}^n a_i$ and $m\geq n$, then 
\begin{align} 
			O_{\{\alpha_1,a_1;\dots,\alpha_n,a_n\}} \big( v^{\beta_1,b_1}\cdots v^{\beta_m,b_m} \big) = \delta_{nm} \sum_{\sigma\in S_n} \prod_{i=1}^n \delta_{a_ib_{\sigma(i)}} \prod_{i=1}^n \prod_{j=0}^{a_i} \partial_x 
			\ddd{\qt^{\gamma_{i,j}}_0\qt_{\gamma_{i,j+1},0}},
		\end{align}  
		where $\gamma_{i,0} = \beta_{\sigma(i)}$ and $\gamma_{i,a_i} = \alpha_i$. 
\end{corollary} 

\begin{proof} Using Corollary~\ref{cor:LeibnizRule-O-multi} we see that 
\begin{align} \label{eq:proof-Cor-m-geq-2-Ooper}
	O_{\{\alpha_1,a_1;\dots,\alpha_n,a_n\}} \big( v^{\beta_1,b_1}\cdots v^{\beta_m,b_m} \big) = \sum_{\substack{I_1\sqcup \cdots \sqcup I_m  \\ 
			= \{1,\dots,n\}}	} \prod_{j=1}^m O_{\{\alpha_i,a_i\}_{i\in I_j}}(v^{\beta_j,b_j}).
\end{align}
Note that by Corollary~\ref{cor:Thm-Ooper-m1} the factor $O_{\{\alpha_i,a_i\}_{i\in I_j}}(v^{\beta_j,b_j})$ vanishes unless $b_j\geq \sum_{i\in I_j} a_i$. But if this inequality holds for every $j=1,\dots,m$, then  $\sum_{i=1}^m b_i\geq \sum_{i=1}^n a_i$. Thus, under the assumption $\sum_{i=1}^m b_i\leq \sum_{i=1}^n a_i-1$ each summand on the right hand side of Equation~\eqref{eq:proof-Cor-m-geq-2-Ooper} has at least one vanishing factor, which implies the first statement of the Corollary. 

For the second statement notice that if $\sum_{i=1}^m b_i=\sum_{i=1}^n a_i$ then $b_j\geq \sum_{i\in I_j} a_i$ for each $j=1,\dots,m$ implies $b_j=\sum_{i\in I_j} a_i$ (otherwise at least one factor in the corresponding summand vanishes) and, therefore, $n\geq m$ (since each $I_j$ must be nonempty). Hence $n\geq m$, and since $m\geq n$ by assumption, we obtain $m=n$. Hence, each $I_j$ consists just of one element that we denote by $a_{\sigma^{-1}(j)}$ for some $\sigma\in S_m$, and the corresponding summand is nonzero only if $a_{\sigma^{-1}(j)} = b_j$, $j=1,\dots,m$. This implies 
\begin{align} 
	O_{\{\alpha_1,a_1;\dots,\alpha_n,a_n\}} \big( v^{\beta_1,b_1}\cdots v^{\beta_m,b_m} \big) = \delta_{nm} \sum_{\sigma\in S_n} \prod_{i=1}^n \delta_{a_ib_{\sigma(i)}} \prod_{i=1}^n O_{\alpha_i,a_i} (v^{\beta_{\sigma(i)},a_i}),
\end{align}  
and we complete the argument by applying Corollary~\ref{cor:Thm-Ooper-m1} to each factor on the right-hand side of this expression.
\end{proof}

	\subsection{Universal identities}
With the properties of the operators given in previous subsections, we are ready to prove Theorem~\ref{af}. In what follows we fix a semisimple Frobenius manifold $M$ (and recall all standard notation as in Section~\ref{al}). First let us derive the action of those operators in terms of the canonical coordinates.
\begin{lemma}
	\label{bj}
	In terms of the canonical coordinates, the genus zero 3-point correlators are given by
	\[
		\ddd{\qt_{\qa,0}\qt_{\qb,0}\qt_{\qg,0}} = \sum_i\frac{\Psi_{i\qa}\Psi_{i\qb}\Psi_{i\qg}}{\Psi_{i1}}u^{i,1}.
		\]
\end{lemma}
\begin{proof}
This lemma is most simply proved by using the tau-structure of the Principal Hierarchy associated with $M$. Alternatively, it follows from tameness of $\mathcal F_0$ that (see \cite{buryak2012polynomial})
\[
	\ddd{\qt_{\qa,0}\qt_{\qb,0}} = \frac{\qp^2F}{\qp v^{\qa}\qp v^{\qb}},
	\]  
here $F$ is the Frobenius potential of $M$ given by
\[
	F = \mathcal F_0|_{t^{\qa,0} = v^{\qa},t^{\qa,1} = t^{\qa,2}=\dots=0}.
	\]
Note that by definition $v^\qa=
\ddd{\qt^\qa_0\qt_{1,0}}$, hence
\[
	\diff{v^\qa}{t^{\qb,0}} = \diff{
	\ddd{\qt^\qa_0\qt_{\qb,0}}}{t^{1,0}} = \qp_\qg\ddd{\qt^\qa_0\qt_{\qb,0}}v^{\qg,1} = c^\qa_{\qb\qg}v^{\qg,1}.
	\]
Therefore, it follows that
\[
	\ddd{\qt_{\qa,0}\qt_{\qb,0}\qt_{\qg,0}} = \qp_\mu
	\ddd{\qt_{\qa,0}\qt_{\qb,0}}
	\diff{v^\mu}{t^{\qg,0}} = c_{\qa\qb\mu}c^\mu_{\qg\qd}v^{\qd,1}.
	\]
The lemma then follows from \eqref{bd}, \eqref{bi}. 
\end{proof}
We can now prove Theorem~\ref{af} with the help of the above lemma. 

\begin{theorem}[=Theorem~\ref{af}]
	Fix an operator $O_{\{\alpha_1,k_1;\dots,\alpha_n,k_n\}}$ with $k_1+\dots+k_n = 3g-3+n$ with $g=n=k_1=1$ or
	\[
		k_1+\cdots +k_n =   3g-3 +n,\quad g\geq 2,\quad n\geq 1,\quad k_i\geq 2,
		\]
	and denote by $\mu$ the partition of $3g-3+n$ given by $k_1,\dots,k_n$.
	Then we have the following universal identity for $\mathcal F_g$ with $g\geq 1$:
	\begin{equation}
		\label{bl}
		O_{\{\alpha_1,k_1;\dots,\alpha_n,k_n\}}(\mathcal F_g) = |\mathrm{Aut}(\mu)|C_{g;\mu}M[g]^{\qg_0}_{\qg_n}\prod_{i=1}^n\ddd{\qt_{\qa_i,0}\qt_{\qg_{i-1,0}}\qt^{\qg_{i},0}},
	\end{equation}
	here $C_{g;\mu}$ is the constant defined in \eqref{ar} and $M[g]$ is defined by 
\[
	M[g]^\qa_\qg =\begin{cases}
		\qd^\qa_\qg, & \text{for } g=1, \\
		M^\qa_\qg, & \text{for } g=2, \\
		M^{\qa}_{\qb_1}M^{\qb_1}_{\qb_2}\dots M^{\qb_{g-3}}_{\qb_{g-2}}M^{\qb_{g-2}}_{\qg}, & \text{for } g\geq 3,
	  \end{cases}
	\]
where we denote 
\[
M^{\qa}_\qb = \ddd{\qt_{0}^\qa\qt_{\ql,0}\qt_{\mu,0}}\ddd{\qt_{0}^\ql\qt_{0}^\mu\qt_{\qb,0}}.
	\]
\end{theorem}
\begin{proof}
        In order to apply $O_{\{\alpha_1,k_1;\dots,\alpha_n,k_n\}}$ to $\mathcal F_g$, we recall Theorem~\ref{ae}. Note that the condition~\eqref{eq:RestrictionOnH} implies that $H_g$ is represented as a linear combination of the monomials $\prod_{j=1}^m u^{i_j,b_j}$ with $b_1+\cdots+b_m< 3g-3+m$, $b_j\geq 2$, with the coefficients being some functions in $u^i$ and $u^{i,1}$, $i=1,\dots,N$. Alternatively, in flat coordinates, $H_g$ is represented as a linear combination of the monomials $\prod_{j=1}^m v^{\alpha_j,b_j}$ with $b_1+\cdots+b_m< 3g-3+m$, $b_j\geq 2$, with the coefficients being some functions in $v^\alpha$ and $v^{\alpha,1}$, $i=1,\dots,N$.
        By the first statement of Theorem~\ref{thm:properties-of-operators-O} the operator $O_{\{\alpha_1,k_1;\dots,\alpha_n,k_n\}}$ commutes with the coefficients that are functions in $v^\alpha$ and $v^{\alpha,1}$ and by the second statement of Theorem~\ref{thm:properties-of-operators-O} (in combination with the iterative property \eqref{bk} of the operator) the operator $O_{\{\alpha_1,k_1;\dots,\alpha_n,k_n\}}$ vanishes the monomials $\prod_{j=1}^m v^{\alpha_j,b_j}$. Thus, $O_{\{\alpha_1,k_1;\dots,\alpha_n,k_n\}} H_g=0$.
        Therefore, we only need to write down the action of $O_{\{\alpha_1,k_1;\dots,\alpha_n,k_n\}}$ on the first summand in~\eqref{ah}.

    To this end, recall that it follows from Lemma~\ref{bj} and Corollary~\ref{cor:Thm-Ooper-m1} that 
	\begin{align} \label{eq:O-flat-canonical}		
    O_{\{\qa,k\}}(u^{i,k}) = \frac{\Psi_{i\qa}}{\Psi_{i1}}(u^{i,1})^{k+1}.
		\end{align}
Also recall the structure of $F_g^{KdV}$ prescribed by Equation \eqref{ar}. Hence, using once again the iterative property \eqref{bk} of the operator, we see that 
	\begin{align*}
O_{\{\alpha_1,k_1;\dots,\alpha_n,k_n\}}(\mathcal F_g) & = 
\sum_{i} f_i(u)^{1-g} O_{\{\alpha_1,k_1;\dots,\alpha_n,k_n\}}( F^{KdV}_g(u^{i,1},\dots,u^{i,3g-2})) \\
& = \sum_{i} \Psi_{i1}^{2-2g} (u^{i,1})^{-g-n+1} \sum_{n\geq 0}\sum_{\nu\in P(g,n)} C_{g;\nu} O_{\{\alpha_1,k_1;\dots,\alpha_n,k_n\}} (v^{i,(\nu)}) 
\\
& = |\mathrm{Aut}(\mu)|C_{g;\mu} \sum_{i} \Psi_{i1}^{2-2g} (u^{i,1})^{-g-n+1} \prod_{j=1}^n O_{\{\qa_j,k_j\}}(u^{i,k_j}).
\\
& = |\mathrm{Aut}(\mu)|C_{g;\mu} \sum_{i} \Psi_{i1}^{2-2g} (u^{i,1})^{-g-n+1} \prod_{j=1}^n \frac{\Psi_{i\qa_j}}{\Psi_{i1}}(u^{i,1})^{k_j+1}.
\\
& = |\mathrm{Aut}(\mu)|C_{g;\mu}\sum_i \frac{\Psi_{i\qa_1}\dots\Psi_{i\qa_n}}{\Psi_{i1}^{2g-2+n}}(u^{i,1})^{2g-2+n}
	\end{align*}
(in the first line we use the vanishing of the operator on $H_g$; in the second line we use Equation \eqref{ar} for $F_g^{KdV}$; in the third line we use \eqref{bk}; in the fourth line we use \eqref{eq:O-flat-canonical}).
    
	By a straightforward computation using Lemma~\ref{bj}, it follows that
	\[
		\prod_{i=1}^n\ddd{\qt_{\qa_i,0}\qt_{\qg_{i-1,0}}\qt^{\qg_{i},0}} = \sum_i\frac{\Psi_{i\qa_1}\dots\Psi_{i\qa_n}\Psi_{i\qg_0}\Psi^{\qg_n}_i}{\Psi_{i1}^n}(u^{i,1})^n,\quad M[g]^\qa_\qg = \sum_i\frac{\Psi_i^\qa\Psi_{i\qg}}{\Psi_{i1}^{2g-2}}(u^{i,1})^{2g-2},
		\]
	which implies the validity of \eqref{bl}. The theorem is proved.
\end{proof}

	\subsection{Towards more general universal identities}
	\label{bz}
	As we have introduced in Sect.\,\ref{ak}, the simplest example of universal identities derived from Theorem~\ref{af} is 
	\[
	\ddg{\qt_{\qa,1}}_1-\ddd{\qt_{\qa,0}\qt^\qb_0}\ddg{\qt_{\qb,0}}_1 = \frac{1}{24}\ddd{\qt_{\qa,0}\qt^\qb_0\qt_{\qb,0}}.
		\]
	However, using Theorem~\ref{ae}, it is immediate to derive the general form of the above genus $1$ recursion relation \cite{dijkgraaf1990mean}:
	\begin{equation}
		\label{bm}
		\ddg{\qt_{\qa,p}}_1-\ddd{\qt_{\qa,p-1}\qt^\qb_0}\ddg{\qt_{\qb,0}}_1 = \frac{1}{24}\ddd{\qt_{\qa,p-1}\qt^\qb_0\qt_{\qb,0}}.
	\end{equation}
	Indeed, let us define the vector field
	\[
		A^0_{\qa,p} = \diff{}{t^{\qa,p}}-\ddd{\qt_{\qa,p-1}\qt^\qb_0}\diff{}{t^{\qb,0}},\quad p\geq 1,
		\]
	then it follows from genus zero topological recursion relation \cite{witten1990two} that 
	\[
		A^0_{\qa,p}(v^\qb) = 0,\quad A^0_{\qa,p}(v^{\qb,1}) =\ddd{\qt_{\qa,p-1}\qt_{\mu,0}\qt_{1,0}}\ddd{\qt^\mu_0\qt_0^\qb\qt_{1,0}}.
		\]
	By using the decomposition \eqref{ah} we have
	\begin{align*}
		A^0_{\qa,p}(\mathcal F_1)&=\sum_i\diff{F_1}{u^{i,1}}\frac{\Psi_{i\qb}}{\Psi_{i1}}\ddd{\qt_{\qa,p-1}\qt_{\mu,0}\qt_{1,0}}\ddd{\qt^\mu_0\qt_0^\qb\qt_{1,0}}\\
		&=\frac{1}{24}\kk{\sum_i\frac{1}{u^{i,1}}\frac{\Psi_{i\qb}}{\Psi_{i1}}}\kk{\sum_j\qp_\qe h_{\qa,p-1}\Psi_j^\qe\Psi_{j,\mu}u^{j,1}}\kk{\sum_\ell \Psi_\ell^\mu\Psi^\qb_\ell u^{\ell,1}}\\
		&=\frac{1}{24}\sum_i\qp_\qe h_{\qa,p-1}\frac{\Psi_{i\qe}}{\Psi_{i1}}u^{i,1}\\
		&=\frac{1}{24}\ddd{\qt_{\qa,p-1}\qt^\qb_0\qt_{\qb,0}},
	\end{align*}
	here in the computation we denote by $h_{\qa,p}\coloneqq \ddd{\qt_{\qa,p}\qt_{1,0}}$, and we use the fact that
	\begin{equation}
		\label{bn}
		\ddd{\qt_{\qa,p}\qt_{\qb,0}\qt_{\qg,0}} = \sum_i \qp_\qe h_{\qa,p}\frac{\Psi_i^\qe\Psi_{i\qb}\Psi_{i\qg}}{\Psi_{i1}}u^{i,1}.
	\end{equation}
The identity \eqref{bn} can be proved similarly as  Lemma~\ref{bj}.

This idea can be used to derive more general universal identities. In this section, we construct some operators that can be viewed as certain generalizations of those operators given in Sect.\,\ref{bq}. 
\begin{definition}
	Define vector fields $A^m_{\qa,p}$ for $p\geq m+1$ and $m\geq 0$ by 
	\[
		A^m_{\qa,p} = \begin{cases}
			O_{\qa,p}, & \text{for}\quad p = m+1,\\
			O_{\qa,p}+\sum_{k = 0}^{p-m-2} 
			 \ddd{\qt^\qg_0\qt_{\qa,k}}O_{\qg,p-k-1},& \text{for}\quad p \geq m+2,
		\end{cases}
		\]
	here $O_{\qa,p}$ are the Eguchi-Xiong operators defined by \eqref{bo}, \eqref{bp}.
\end{definition}
We have the following observations on these operators.
\begin{lemma}
	The vector fields $A^m_{\qa,p}$ are in involution.
\end{lemma}
\begin{proof}
	This follows from the fact that $[O_{\qa,p},O_{\qb,q}]=0$ for $p,q\geq 1$ and $O_{\qa,p}(v^\qb) = 0$ for $p\geq 1$. The lemma is proved.
\end{proof}
\begin{lemma}
	The vector fields $A^m_{\qa,p}$ satisfy the recursion relation
	\begin{align}
		\label{br}
		A^0_{\qa,p} &= \diff{}{t^{\qa,p}}-
		\ddd{\qt^{\qg}_0\qt_{\qa,p-1}}
		\diff{}{t^{\qg,0}},\\
		\label{bs}
		A^{m+1}_{\qa,p} &= A^{m}_{\qa,p}-
		\ddd{\qt^{\qg}_0\qt_{\qa,p-m-2}}
		A^m_{\qg,m+1},\quad m\geq 0.
	\end{align}
\end{lemma}
\begin{proof}
The lemma is proved directly from the definition.
\end{proof}
The recursive description of $A^m_{\qa,p}$ allows us to compute their action in terms of the flat coordinates.
\begin{proposition}
	We have
	\begin{align}
		\label{bt}
		A^{m}_{\qa,p}(v^{\qb,r})&=
			\sum_{j_1=0}^{r-1}\sum_{j_2=0}^{j_1-1}\dots\sum_{j_{m+1}=0}^{j_m-1}\binom{r}{j_1+1}\binom{j_1}{j_2+1}\dots\binom{j_m}{j_{m+1}+1}\prod_{\ell = 1}^m\ddd{\qt^{\qg_\ell}_0\qt_{\qg_{\ell+1},0}\qt_{1,0}^{j_{m+1-\ell}-j_{m+2-\ell}}}\\
			\notag
			&\times\ddd{\qt_{\qa,p-m-1}\qt_{\qg_1,0}\qt_{1,0}^{j_{m+1}+1}}\ddd{\qt_0^{\qg_{m+1}}\qt_{0}^\qb\qt_{1,0}^{r-j_1}}.
	\end{align}
	In particular, we have $A^{m}_{\qa,p}(v^{\qb,r}) = 0$ for $r\leq m$.
\end{proposition} 
\begin{proof}
 For $A^{0}_{\qa,p}$, we have
	\begin{align*}
		A^{0}_{\qa,p}(v^{\qb,r}) &= \ddd{\qt_{\qa,p}\qt^\qb_0\qt_{1,0}^{r+1}}-\ddd{\qt_{\qa,p-2}\qt^\mu}\ddd{\qt_{\mu,0}\qt^\qb_0\qt_{1,0}^{r+1}}\\
		&=\sum_{j=0}^{r-1}\binom{r}{j+1}\ddd{\qt_{\qa,p-1}\qt_{\mu,0}\qt_{1,0}^{j+1}}\ddd{\qt^\mu_0\qt_{0}^\qb\qt_{1,0}^{r-j}},
	\end{align*}
	here we use the genus zero topological recursion relation to derive the second line. Hence, \eqref{bt} holds true for $m=0$. Then it is straightforward to prove the general case by induction on $m$ using the recursive relation \eqref{bs}. The proposition is proved. 
\end{proof}
	
\begin{corollary}
	\label{bu}
	In terms of the canonical coordinates, we have
	\[
		A^{m}_{\qa,p}(v^{\qb,m+1}) = \sum_i\qp_\qe h_{\qa,p-m-1}\Psi_{i}^\qe\Psi_i^\qb (u^{i,1})^{m+2}.
		\]
\end{corollary}
\begin{proof}
	Taking $r=m+1$ in the identity \eqref{bt}, we have
	\[
		A^{m}_{\qa,p}(v^{\qb,m+1}) = \ddd{\qt_{\qa,p-m-1}\qt_{\qg_1,0}\qt_{1,0}}\ddd{\qt_0^{\qg_{m+1}}\qt^\qb_0\qt_{1,0}}\prod_{k=1}^m\ddd{\qt_0^{\qg_k}\qt_{\qg_{k+1},0}\qt_{1,0}}.
		\]
	Then we prove the corollary by a straightforward computation using \eqref{bn}. 
\end{proof}

We continue to construct universal identities using the operators $A^m_{\qa,p}$. It is straightforward to obtain the following result.
\begin{proposition}
	For any $p\geq 3g-2$, we have a universal identity given by
	\begin{equation}
		\label{bv}
		A^{3g-3}_{\qa,p}(\mathcal F_g) = \ddd{\qt_{\qa,p-3g-2}\qt^\qb_0\qt_{\mu,0}}M[g]^\mu_\qb\int_{\oM_{g,1}}\psi_1^{3g-2}.
	\end{equation}
\end{proposition}
\begin{proof}
	Combining Theorem~\ref{ae} and Corollary~\ref{bu}, it follows that
	\[
		A^{3g-3}_{\qa,p}(\mathcal F_g) = \sum_i\qp_\qe h_{\qa,p-3g-2}\frac{\Psi_i^\qe}{\Psi_{i1}^{2g-1}}(u^{i,1})^{2g-1}\int_{\oM_{g,1}}\psi_1^{3g-2}.
		\]
		Using \eqref{bn} and the expression for $M[g]$, the right-hand side of \eqref{bv} reads
		\[
			\text{R.H.S.} = \kk{\sum_i \qp_\qe h_{\qa,p-3g-2}\frac{\Psi_i^\qe\Psi_{i}^\qb\Psi_{i\mu}}{\Psi_{i1}}u^{i,1}}\kk{\sum_j\frac{\Psi_j^\mu\Psi_{j\qb}}{\Psi_{j1}^{2g-2}}(u^{j,1})^{2g-2}}.
			\]
		The proposition is proved.
\end{proof}
It is clear that the identity \eqref{bv} is the general form of \eqref{bl} for $n=1$ and $k_1 = 3g-2$. However, it is not easy to generalize this for $n\geq 2$. 
\begin{definition}
	Define $A^{m_1,\dots,m_n}_{\qa_1,p_1;\dots;\qa_n,p_n}$ to be the order $n$ differential operator
	\[
		A^{m_1,\dots,m_n}_{\qa_1,p_1;\dots;\qa_n,p_n} = A^{m_1}_{\qa_1,p_1}\comp\dots\comp A^{m_n}_{\qa_n,p_n}.
		\]
\end{definition}

Generally speaking, the action of operator $A^{m_1,\dots,m_n}_{\qa_1,p_1;\dots;\qa_n,p_n}$ is hard to describe. We have the following properties.
\begin{proposition}
	\label{bw}
	We have
	\[
		A^{m_1,\dots,m_n}_{\qa_1,p_1;\dots;\qa_n,p_n}(v^{\qb,r}) = 0,\quad r\leq m_1+\dots+m_n.
		\]
\end{proposition}
\begin{proof}
	Note that as a differential polynomial, genus zero $n$-point functions are of differential degree $n-2$ for $n\geq 2$. In particular, 
	\[
\ddd{\qt_{\qa,p}\qt_{\qb,q}\qt_{1,0}^j}\in C^\infty(v)[v^{\qg,1},\dots,v^{\qg,j}],\quad j\geq 1.
		\]
	Hence, it follows from \eqref{bt} that 
	\[
		A^m_{\qa,p}(v^{\qb,r})\in C^\infty(v)[v^{\qg,1},\dots,v^{\qg,r-m}].
		\]
	Then it is easy to see that
	\[
		A^{m_1,\dots,m_n}_{\qa_1,p_1;\dots;\qa_n,p_n}(v^{\qb,r})\in C^\infty(v)[v^{\qg,1},\dots,v^{\qg,r-m_1-\dots-m_n}].
		\]
	If $A^{m_1,\dots,m_n}_{\qa_1,p_1;\dots;\qa_n,p_n}(v^{\qb,r})$ is non-zero, then it is a differential polynomial of degree $r+n$ and hence
	\[
		r-m_1-\dots-m_n\geq 1.
		\]
The proposition is proved.
\end{proof}
\begin{corollary}
	\label{bx}
	Fix $m_1,\dots,m_{\ell}$, $r_1,\dots,r_n$, with $m_i\geq 1$, $r_i\geq 2$ and $m_1+\dots+m_{\ell} = 3g-3$ for some $g\geq 2$.
	Then we have 
	\[
		A^{m_1,\dots,m_n}_{\qa_1,p_1;\dots;\qa_\ell,p_\ell}(v^{\qb_1,r_1}\dots v^{\qb_n,r_n}) = 0,\quad r_1+\dots+r_n\leq 3g-4+n.
		\]
\end{corollary}
\begin{proof}
	By successively using the Leibniz rule, the action \[A^{m_1,\dots,m_\ell}_{\qa_1,p_1;\dots;\qa_\ell,p_\ell}(v^{\qb_1,r_1}\dots v^{\qb_n,r_n})\] can be computed by distributing the action of vector fields $A^{m_i}_{\qa_i,p_i}$ on each $v^{\qb_j,r_j}$. If the action is to give a non-zero result, without loss of generality, we can assume that the following expression is non-zero:
	\[
		A^{m_1,\dots,m_{\ell_1}}_{\qa_1,p_1;\dots;\qa_{\ell_1},p_{\ell_1}}(v^{\qb_1,r_1})\dots A^{m_{\ell_{s-1}+1},\dots,m_{\ell_s}}_{\qa_{\ell_{s-1}+1},p_{\ell_{s-1}+1};\dots;\qa_{\ell_s},p_{\ell_s}}(v^{\qb_s,r_s})v^{\qb_{s+1},r_{s+1}}\dots v^{\qb_n,r_n},
		\]
	for some $s\leq n$. Then it follows from Proposition~\ref{bw} that 
	\[
		r_1\geq m_1+\dots+m_{\ell_1}+1,\quad\dots,\quad r_s\geq m_{\ell_{s-1}+1}+\dots+m_\ell+1.
		\]
	Therefore, the action is non-zero unless
	\[
		r_1+\dots+r_n\geq m_1+\dots+m_\ell+s+2(n-s)\geq 3g-3+n.
		\]
		The corollary is proved.
\end{proof}
The above corollary implies that, with the help of the decomposition \eqref{ah}, the action
\[
	A^{m_1,\dots,m_n}_{\qa_1,p_1;\dots;\qa_n,p_n}(v^{\qb,r})(\mathcal F_g),\quad m_1+\dots+m_n = 3g-3,\quad m_i\geq 1,\quad g\geq 2
	\]
can be represented in terms of genus zero correlators, and hence produce universal identities that can be viewed as general forms of \eqref{bl}. However, it is not straightforward to write down the explicit forms of these identities. Let us illustrate the idea by considering $g=n=2$. 

We fix $p_1\geq 2$ and $p_2\geq 3$. Denote by \[
	Q^{m;\qb,r}_{\qa,p} = A^m_{\qa,p}(v^{\qb,r}),
	\]
then it follows from Corollary~\ref{bx} that
\[
	A^{2,1}_{\qa_2,p_2;\qa_1,p_1}(\mathcal F_2) = Q^{2;\qb_2,3}_{\qa_2,p_2}Q^{1;\qb_1,2}_{\qa_1,p_1}\frac{\qp^2 F_2}{\qp v^{\qb_1,2}v^{\qb_2,3}}+Q^{2;\qb_2,3}_{\qa_2,p_2}\diff{Q^{1;\qb_1,4}_{\qa_1,p_1}}{v^{\qb_2,3}}\diff{F_2}{v^{\qb_1,4}}
	\]
By using Corollary~\ref{bu}, we see that 
\begin{equation}
	\label{by}
	Q^{1;\qb_1,2}_{\qa_1,p_1} = \sum_i\qp_\qa h_{\qa_1,p_1-2}\Psi_i^\qe\Psi_i^{\qb_1}(u^{i,1})^3,\quad Q^{2;\qb_2,3}_{\qa_2,p_2} = \sum_i\qp_\qa h_{\qa_2,p_2-3}\Psi_i^\qe\Psi_i^{\qb_2}(u^{i,1})^4,
\end{equation}
and by using \eqref{ah} and the expression for $F_2^{KdV}$ we see that 
\[
	\frac{\qp^2 F_2}{\qp v^{\qb_1,2}v^{\qb_2,3}} = -\frac{7}{1920}\sum_i\frac{\Psi_{i\qb_1}\Psi_{i\qb_2}}{\Psi_{i1}^5}\frac{1}{(u^{i,1})^3}.
	\]
It is a straightforward computation to obtain that
\[
	Q^{2;\qb_2,3}_{\qa_2,p_2}Q^{1;\qb_1,2}_{\qa_1,p_1}\frac{\qp^2 F_2}{\qp v^{\qb_1,2}v^{\qb_2,3}} = -\frac{7}{1920}\ddd{\qt_{\qa_1}\qt_{\qb,0}\qt^\mu_0}\ddd{\qt_{\qa_2,p_2-3}\qt^\qb_0\qt_{\ql,0}}M[2]^\ql_\mu
	\]
As for the second term, it follows from the identity \eqref{bt} that 
\[
	Q^{1;\qb_1,4}_{\qa_1,p_1} = \sum_{j=0}^3\sum_{\ell = 0}^{j-1}\binom{r}{j+1}\binom{j}{\ell+1}\ddd{\qt_{\qa_1,p_1-2}\qt^\qe_0\qt_{1,0}^{\ell+1}}\ddd{\qt_{\qe,0}\qt_{\qg,0}\qt_{1,0}^{j-\ell}}\ddd{\qt^\qg_0\qt_0^{\qb_1}\qt_{1,0}^{4-j}}
	\]
Hence we see that 
\begin{align*}
	\diff{Q^{1;\qb_1,4}_{\qa_1,p_1}}{v^{\qb_2,3}}=&\,6\ddd{\qt_{\qa_1,p_1-2}\qt^\qe_0\qt_{1,0}}\ddd{\qt_{\qe,0}\qt_{\qg,0}\qt_{1,0}}c^{\qg\qb_1}_{\qb_2}+3\ddd{\qt_{\qa_1,p_1-2}\qt^\qe_0\qt_{1,0}}c_{\qe\qg\qb_2}\ddd{\qt^\qg_0\qt_0^{\qb_1}\qt_{1,0}}\\
	&+\qp_{\qb_2}\ddd{\qt^\qe_0\qt_{\qa_1,p_1-2}}\ddd{\qt_{\qe,0}\qt_{\qg,0}\qt_{1,0}}\ddd{\qt^\qg_0\qt_0^{\qb_1}\qt_{1,0}},
\end{align*}
and we have
\[
	Q^{2;\qb_2,3}_{\qa_2,p_2}\diff{Q^{1;\qb_1,4}_{\qa_1,p_1}}{v^{\qb_2,3}}\diff{F_2}{v^{\qb_1,4}} = \frac{5}{576}\ddd{\qt_{\qa_1}\qt_{\qb,0}\qt^\mu_0}\ddd{\qt_{\qa_2,p_2-3}\qt^\qb_0\qt_{\ql,0}}M[2]^\ql_\mu.
	\]
To summarize, we have the following universal identity for $g=2$:
\[
	A^{2,1}_{\qa_2,p_2;\qa_1,p_1}(\mathcal F_2) = \frac{29}{5760}\ddd{\qt_{\qa_1}\qt_{\qb,0}\qt^\mu_0}\ddd{\qt_{\qa_2,p_2-3}\qt^\qb_0\qt_{\ql,0}}M[2]^\ql_\mu.
	\]

\begin{remark}
    \label{revre}
    As a final remark, let us explain more concretely why the operator $\nordbullet\prod_{i=1}^n O_{\alpha_i,p_i}\nordbullet$ defined by simply taking the normal order product is not suitable for our construction of universal identities. We begin by expressing the Eguchi-Xiong operator by 
    \[
    O_{\qa,p} = \sum_{r=0}^p H^{\mu,r}_{\qa,p}\diff{}{t^{\mu,r}} = \sum_{k\geq 0}G^{\qb,k}_{\qa,p}\diff{}{v^{\qb,k}},
    \]
    here the coefficients $ H^{\mu,r}_{\qa,p}$ and $G^{\qb,k}_{\qa,p}$ are certain differential polynomials given by products of genus zero correlation functions. We know that $G^{\qb,k}_{\qa,p} = A^{p-1}_{\qa,p}(v^{\qb,k})$,  hence it vanishes for $k<p$. It is immediate to obtain that 
    \[G^{\qb,k}_{\qa,p} =  \sum_{r=0}^p H^{\mu,r}_{\qa,p}\ddd{\qt_{\mu,r}\qt^\qb_0\qt_{1,0}^{k+1}}.\]
    Now we simply consider the operator $\tilde O = \nordbullet O_{\qa_1,p_1}O_{\qa_2,p_2}\nordbullet$, which reads
    \begin{align*}
     \tilde O=&\,\sum_{\substack{r_1\leq p_1\\r_2\leq p_2}}H^{\mu_1,r_1}_{\qa_1,p_1}H^{\mu_2,r_2}_{\qa_2,p_2}\frac{\qp^2}{\qp t^{\mu_1,r_1}\qp t^{\mu_2,r_2}}\\
    =&\sum_{\substack{r_1\leq p_1\\r_2\leq p_2}}\sum_{k,s\geq 0}H^{\mu_1,r_1}_{\qa_1,p_1}H^{\mu_2,r_2}_{\qa_2,p_2}\left(\ddd{\qt_{\mu_1,r_1}\qt^\qb_0\qt_{1,0}^{k+1}}\diff{}{v^{\qb,k}}\right)\comp\left(\ddd{\qt_{\mu_2,r_2}\qt^\ql_0\qt_{1,0}^{s+1}}\diff{}{v^{\ql,s}}\right)\\
    =&\sum_{r_2\leq p_2}\sum_{k,s\geq 0}G^{\qb,k}_{\qa_1,p_1}H^{\mu_2,r_2}_{\qa_2,p_2}\left(\diff{}{v^{\qb,k}}\ddd{\qt_{\mu_2,r_2}\qt^\ql_0\qt_{1,0}^{s+1}}\right)\diff{}{v^{\ql,s}}\\
    &+\sum_{k,s\geq 0} G^{\qb,k}_{\qa_1,p_1}G^{\ql,s}_{\qa_2,p_2}\frac{\qp^2}{\qp v^{\qb,k}\qp v^{\ql,s}}
    \end{align*}
    Notice that in terms of jet coordinates, the differential operator $\tilde O$ consists of not only second order operators but also first order operators. In particular, the coefficients $G^{\qb,k}_{\qa_1,p_1}G^{\ql,s}_{\qa_2,p_2}$ of the second order part are nonzero if and only if $k\geq p_1$ and $s\geq p_2$, while the coefficients of the first order part are non-vanishing for $k\geq p_1$ and $s+1\geq k$. This makes the construction of universal identity potentially problematic. As an example, we set $p_1=2$ and $p_2 = 3g-3$ for $g\geq 2$. Then we know, by applying Theorem \ref{ae}, that
    \begin{align*}
    \tilde O(F_g) =&\, G^{\qb,2}_{\qa_1,2}G^{\ql,3g-3}_{\qa_2,3g-3}\frac{\qp^2 F_g}{\qp v^{\qb,2}\qp v^{\ql,3g-3}}\\
    &+\sum_{r_2\leq 3g-3}\sum_{\substack{k\geq 2\\s\geq 1}}G^{\qb,k}_{\qa_1,2}H^{\mu_2,r_2}_{\qa_2,3g-3}\left(\diff{}{v^{\qb,k}}\ddd{\qt_{\mu_2,r_2}\qt^\ql_0\qt_{1,0}^{s+1}}\right)\diff{F_g}{v^{\ql,s}}.
    \end{align*}
    It is then obvious that all the gradients $\diff{F_g}{v^{\ql,s}}$ for $s\geq 1$ are involved in such an identity. In particular, we need to compute $\diff{H_g}{v^{\ql,s}}$, of which we do not have any control at the present time. Therefore, we need to modify $\tilde O$ to get rid of terms corresponding to $H_g$, and this is exactly the essential role played by the additional terms appeared in \eqref{aw}.
\end{remark}

\section{Conclusion}\label{ag}
In this paper, we study the universal identities for tau-functions (or, more precisely, the free energy functions) of the Dubrovin-Zhang hierarchies. The result is that we can derive a family of universal identities for each genus $g\geq 1$, and these identities don't seem to follow directly from the known relations among the tautological classes on the moduli spaces of curves. 

Moreover, the identities are derived in a particular way that does clarify the structure of the free energy functions in the Dubrovin-Zhang formalism. It is worth to remark that the decomposition \eqref{ah} should be viewed as decomposing the free energy function $F_g$ into leading terms and lower order terms. It is interesting to ask how to identify in a similar explicit way the next order term in the decomposition \eqref{ah} and derive corresponding universal identities.

In this paper, we combine Givental's quantization formalism and Dubrovin-Zhang's loop equation to derive \eqref{ah}. We remark that this interaction is using the 2005 theorem of Dubrovin and Zhang that is not fully publicly presented yet (but available in~\cite{dubrovin2005normal}). One of the further research directions that the authors find very important is to use the enormous technical development of this field in the past 20 years to revisit this theorem and make it fully publicly available. 

The key observation is that, viewed as polynomials in $\frac{1}{\ql-u^1},\dots,\frac{1}{\ql-u^N}$, the leading terms of both sides of the loop equation can be explicitly written. However, it is difficult to consider even the next order terms. Therefore, one may ask if the property of $F_g$ can be studied with other approaches. The \textit{polynomiality theorem} \cite{liu2021linearization} may be a good candidate to study $F_g$. Indeed, the relation between the polynomiality theorem and some tautological relations on the moduli space is studied in \cite{iglesias2022bi}. The polynomiality should give some constraints on the form of $F_g$ and by combining with the loop equation, one may find more structures for $F_g$.

It is also interesting to note that the results of this paper can be applied in a wider context than the free energy functions in the Dubrovin-Zhang formalism. To this end, one can consider the partition functions of not necessarily homogeneous semisimple cohomological field theories. Their relation to the Dubrovin-Zhang tau-functions can be described by the following system of observations:
\begin{itemize}
	\item A homogeneous semisimple CohFT determines a formal Frobenius manifold semisimple at the origin, and for a particular choice of calibration the Dubrovin-Zhang tau function coincides with the corresponding CohFT partition function. 
	\item Vice versa,  the formal expansion of a semisimple Frobenius manifold near each its semisimple point determines a homogeneous CohFT in all genera; moreover, by Teleman's result~\cite{teleman2012structure} the homogeneous CohFT in this case is uniquely determined by its genus 0 part. In this case the partition function of thus constructed CohFT is obtained from the Dubrovin-Zhang tau function by a lower triangular element of the Givental group. 
	\item If we drop the assumption of homogeneity for a CohFT, many of the used techniques still work. For instance, their partition functions are still tame and in the same orbit of the Givental group, and they are tau-functions of some Hamiltonian hierarchies ~\cite{buryak2012deformations,buryak2012polynomial}. But, in general, we don't have the second Hamiltonian structure and lose the loop equation.
\end{itemize}
We note that Proposition~\ref{prop:functions-varphi} can also be proved for the partition functions of not necessarily homogeneous semisimple cohomological field theories through the analysis of the Givental formula in terms of graphs, as in~\cite{dunin2013givental}. In particular, in combination with Proposition~\ref{bb} this means that the KdV free energy functions serve as universal leading terms for the partition functions of any semisimple cohomological field theory. However, we don't expect that we can omit analysis through the loop equation for the next order terms. 

Finally, in deriving the universal identities, the operators $O_{\{\qa_1,k_1;\dots,\qa_n,k_n\}}$ play an important role. They are not very straightforward to define (and remarkably their structure is related to the tautological relations responsible for the DR/DZ equivalence conjecture and polynomiality of the conservation laws of DZ hierarchies of more general F-CohFTs, see~\cite{buryak2022tautological}), but they possess nice properties and their actions can be explicitly written. However, as shown in Sect.\,\ref{bz}, those operators are not enough to derive more general universal identities. We propose some operators $A^{m_1,\dots,m_n}_{\qa_1,p_1;\dots;\qa_n,p_n}$ that are more general, and to some extent they do serve the purpose, but it is not yet a fully satisfactory set of operators, since their action is hard to determine explicitly. The operators $A^{m_1,\dots,m_n}_{\qa_1,p_1;\dots;\qa_n,p_n}$ can be interpreted as the leading terms of the operators coming from more general tautological relations on the genus zero moduli space studied in~\cite{buryak2022tautological}, and we hope that this link might help to derive more general universal identities.

\vskip 1em
\noindent \textbf{Acknowledgements} S.~S. was supported by the Netherlands Organization for Scientific Research. Z.~W. is a JSPS International Research Fellow and his research is supported by JSPS KAKENHI Grant Number 23KF0114. Z.~W. would like to thank Korteweg-de Vries Institute, University of Amsterdam for its hospitality where part of the work was carried out. Z.~W. would like to thank Si-Qi Liu and Youjin Zhang for very helpful discussions. The authors thank the anonymous referees for useful suggestions.


\end{document}